\documentclass[twoside,12pt]{article}      % Comments after  % are ignored
\usepackage{amsmath,amsfonts} % Typical maths resource packages
\usepackage[pdftex]{graphicx}
\usepackage{amsthm}
\usepackage[left=1in,top=1in,bottom=1in,right=1in]{geometry}
\usepackage{natbib}
\usepackage{algorithm}
\usepackage{algorithmic}
\usepackage{enumerate}
\usepackage{setspace}
\usepackage{multirow}

\def\keywordname{{\bf Key words:}}

\newcommand{\R}{\mathbb{R}}
\newcommand{\p}{\mathbb{P}}
\newcommand{\E}{\mathbb{E}}
\newcommand{\1}{\mathbf{1}}

\newcommand{\g}{\, | \,}

\newcommand{\x}{\mathbf{x}}
\newcommand{\X}{\mathbf{X}}

\newcommand{\A}{\mathbf{A}}

\newcommand{\bigO}{\mathcal{O}}
\newtheorem{thm}{Theorem}[section]

\newtheorem{lem}[thm]{Lemma}

\newcommand{\keywords}[1]{\par\addvspace\baselineskip
\noindent\keywordname\enspace\ignorespaces#1}

\begin{document}

\title{Multivariate convex regression with adaptive partitioning}

\author{Lauren A. Hannah, David B. Dunson}

\date{\today}

\maketitle

%\doublespacing

\begin{abstract}
We propose a new, nonparametric method for multivariate regression subject to convexity or concavity constraints on the response function.  Convexity constraints are common in economics, statistics, operations research, financial engineering and optimization, but there is currently no multivariate method that is computationally feasible for more than a few hundred observations.  We introduce Convex Adaptive Partitioning (CAP), which creates a globally convex regression model from locally linear estimates fit on adaptively selected covariate partitions. %Adaptive partitioning makes computation efficient even on large problems. Convexity itself acts as a regularizer, making CAP resistant to overfitting. 
CAP is computationally efficient, in stark contrast to current methods. The most popular method, the least squares estimator, has a computational complexity of $\bigO(n^3)$. We show that CAP has a computational complexity of $\bigO(n \log(n)\log(\log(n)))$ and also give consistency results. CAP is applied to value function approximation for pricing American basket options with a large number of underlying assets.
\end{abstract}

\section{Introduction}

Consider the regression model for $\x \in \mathcal{X} \subset \R^p$ and $y \in \R,$
$$y = f_0(\x) + \epsilon,$$
where $f_0:\R^p \rightarrow \R$ and $\epsilon$ is a mean 0 random variable.  In this paper, we study the situation where $f_0$ is subject to a convexity constraint.  That is, 
\begin{equation}\notag
\lambda f_0(\x_1) + (1-\lambda)f_0(\x_2) \geq f_0(\lambda \x_1 + (1-\lambda)\x_2),
\end{equation}for every $\x_1,\x_2 \in \mathcal{X}$ and $\lambda \in (0,1)$. Given the observations $(\x_1,y_1),\dots,(\x_n,y_n)$, we would like to estimate $f_0$ subject to the convexity constraint; this is called the convex regression problem.  Note that convex regression is easily extended to concave regression since a concave function is the negative of a convex function.

Convex regression problems occur in a variety of settings. Economic theory dictates that demand~\citep{Va82}, production~\citep{Va84,AlBeHa07} and consumer preference~\citep{BoVa04} functions are often concave. In financial engineering, stock option prices usually have convexity restrictions~\citep{AiDu03}. Stochastic optimization problems, studied in operations research and reinforcement learning, have response surfaces~\citep{Li10} or value-to-go functions that exhibit concavity in many settings, like resource allocation~\citep{ToPo03,Po07,ToNeSa10} or stochastic control~\citep{KeWaBo11}. Similarly, efficient frontier methods like data envelopment analysis~\citep{KuJo10} include convexity constraints. In statistics, shape restrictions like log-concavity are useful in density estimation~\citep{CuSaSt10,CuSa10,ScDu10}. Finally, in optimization, convex approximations to posynomial constraints are valuable for geometric programming~\citep{KiLeVa04,BoKiVa07,MaBo09}. Although convex regression has been well-explored in the univariate setting, the literature remains underdeveloped in the multivariate setting. Existing methods do not scale well to more than a few thousand observations or more than a handful of dimensions.

In this paper, we introduce the first computationally efficient, theoretically sound multivariate convex regression method, called Convex Adaptive Partitioning (CAP). It relies on an alternate definition of convexity,
\begin{equation}\label{eq:convexity}
f_0(\x_1) \geq f_0(\x_2) + g_0(\x_1)^T (\x_1-\x_2),
\end{equation} for every $\x_1,\x_2 \in \mathcal{X}$, where $g_0(\x) \in \partial f_0(\x)$ is a subgradient of $f_0$ at $\x$.  Equation (\ref{eq:convexity}) states that a convex function lies above all of its supporting hyperplanes, or subgradients tangent to $f_0$. Moreover, with enough supporting hyperplanes, $f_0$ can be approximately reconstructed by taking the maximum over those hyperplanes.

The CAP estimator is formed by adaptively partitioning a set of observations. Within each subset of the partition, we fit a linear model to approximate the subgradient of $f_0$ within that subset. Given a partition with $K$ subsets and linear models, $(\alpha_k,\beta_{k})_{k=1}^K$, a continuous, convex (concave) function is then generated by taking the maximum (minimum) over the hyperplanes by
\begin{align}\notag
\hat{f}_{n}(\x) & = \max_{k \in \{1,\dots,K\} }\alpha_k + \beta_k^T \x.
\end{align}
The partition is refined by a twofold strategy. First, one of the subsets is split along a cardinal direction (say, $x_1$ or $x_3$) to grow $K$. Then, the hyperplanes themselves are used to refit the subsets. A piecewise linear function like $\hat{f}_n$ induces a partition; a subset is defined as the region where a particular hyperplane is dominant. The refitting step places the hyperplanes in closer alignment with the observations that generated them. This procedure is repeated until all subsets have a minimal number of observations. The CAP estimator is then created by selecting the value of $K$ that balances fit with complexity using a generalized cross validation method~\citep{GoHeWa79,Fr91}. %Conceptually, CAP is similar to existing regression tree models~\citep{BrFrOl84,ChHuLo94,No96,GyKoKr02}, although the partitioning and model selection methods are quite different.

% Theoretical stuff:
% - computational complexity:
% 	~ us: O((p+1)^2 n x p log(n)*log(log(n))) (LS = O((p+1)^2n), we do this log(n)*log(log(n)) (O(log[log(n)!]))(?) times)
% 	~ LSE: O(n^5) ??? (n  = (p+1)n; LP = O(n^3 (p+1)^3 L), L ~ 2(p+1)n^2
% - consistency (supremum norm)

CAP has strong theoretical properties, both in terms of computational complexity and asymptotic properties. We show that CAP is consistent with respect to the $\ell_{\infty}$ metric and has a computational complexity of $\bigO(p(p+1)^2 n \log(n)\log( \log(n)))$ flops. The most widely implemented convex regression method, the least squares estimator, has only recently been shown to be consistent~\citep{SeSe11,LiGl11} and has a computational complexity of $\bigO((p+1)^3n^3)$ flops. Despite a difference of almost $\bigO(n^2)$ runtime, the CAP estimator usually has better predictive error as well. Because of its dramatic reduction in runtime, CAP opens a new class of problems for study, namely moderate to large problems with convexity or concavity constraints.

The rest of this paper is organized as follows. In Section \ref{sec:litReview}, we review the literature on convex regression. In Section \ref{sec:CAP}, we present the CAP algorithm. In Section \ref{sec:theory}, we give computational complexity results and conditions for consistency. In Section \ref{sec:implementation}, we derive a generalized cross-validation method and give a fast approximation for the full CAP algorithm. In Section \ref{sec:numbers}, we empirically test CAP on convex regression problems, including value function estimation for pricing American basket options. In Section \ref{sec:conclusions}, we discuss our results and give directions for future work.

\section{Literature Review}\label{sec:litReview}

% Univariate literature:

%  - least squares: Hi54 (original), De73 (math programming formulation), Wu82, Dy83, FrMa89 (algorithmic), HaPl76 (consistency), Ma91 (rate of convergence), GrJoWe01 (asymptotic distribution)
% - constrained kernel methods: BiDe06
% - splines: Me08 (convex splines), MeHaHo11 (Bayesian convex splines), ShWaDa11 (Bayesian constrained parameters on splines), Tu05 (frequentist constrained parameters)
% - Bernstein polynomials: ChChHs07 (Bayesian)
% - combinatorial optimization: KoMaPo10

% Isotonic:

% - maximum likelihood: Br55
% - constrained kernel methods: HaHu01
% - splines: ShSaWa09
% - piecewise linear: NeDu04 (Bayesian)

The literature for nonparametric convex regression is dispersed over a variety of fields, including statistics, operations research, economics, numerical analysis and electrical engineering. There seems to be little communication between the fields, leading to the independent discovery of similar techniques.

In the univariate setting, there are many computationally efficient algorithms for convex regression. These methods rely on the ordering implicit to the real line. Setting $x_{i-1} < x_i < x_{i+1}$ for $i = 2,\dots, n$,
\begin{align}\label{eq:ordering}
\frac{f_0(x_i) - f_0(x_{i-1})}{x_i - x_{i-1}} & \leq \frac{f_0(x_{i+1}) - f_0(x_i)}{x_{i+1} - x_i}, & i & = 2,\dots, n, 
\end{align} is a sufficient constraint for pointwise convexity. When $f_0$ is differentiable, Equation (\ref{eq:ordering}) is equivalent to an increasing derivative function. 

Various methods have been used to solve the univariate convex regression problem. The least squares estimator (LSE) is the oldest and simplest method. It produces a piecewise linear estimator by solving a quadratic program with $n-2$ linear constraints~\citep{Hi54,De73}. Although the LSE is completely free of tunable parameters, the estimator is not smooth and can overfit, particularly in the multivariate setting. Consistency, rate of convergence, and asymptotic distribution of the LSE were shown by \citet{HaPl76}, \citet{Ma91} and \citet{GrJoWe01}, respectively. Algorithmic methods for solving the quadratic program were given in \citet{Wu82,Dy83} and \citet{FrMa89}.

Spline methods have also been popular. \citet{Me08} and \citet{MeHaHo11} used convex-restricted splines with positive parameters in frequentist and Bayesian settings, respectively. \citet{Tu05} and \citet{ShWaDa11} used unrestricted splines with restricted parameters, likewise, in frequentist and Bayesian settings. In other methods, \citet{BiDe07} used convexity constrained kernel regression; \citet{ChChHs07} used a random Bernstein polynomial prior with constrained parameters; and \citet{KoMaPo10} transformed the ordering problem into a combinatorial optimization problem which they solved with dynamic programming.

Due to the constraint on the derivative of $f_0$, univariate convex regression is quite similar to univariate isotonic regression. The latter has been studied extensively with many approaches; for examples, see \citet{Br55,HaHu01,NeDu04} and \citet{ShSaWa09}.

Unlike the univariate setting, convex functions in multiple dimensions cannot be represented by a simple set of first order conditions and projection onto the set of convex functions becomes computationally intensive. As in the univariate case, the earliest and most popular regression method is the LSE, which directly projects a least squares estimator onto the cone of convex functions. It was introduced by \citet{Hi54} and \citet{Ho79}. The estimator is found by solving the quadratic program,
\begin{align}\label{eq:lse}
\min & \sum_{i=1}^n \left(y_i - \hat{y}_i\right)^2 \\\notag
\mathrm{subject \ to \ } & \hat{y}_j \geq \hat{y}_i + \mathbf{g}_i^T(\x_j - \x_i), \ \ \ i,j = 1,\dots,n.
\end{align}Here, $\hat{y}_i$  and $\mathbf{g}_i$ are the estimated values of $f_0(\x_i)$ and the subgradient of $f_0$ at $\x_i$, respectively. The estimator $\hat{f}_n^{LSE}$ is piecewise linear,
\begin{equation}\notag
\hat{f}_n^{LSE}(\x) = \max_{i \in \{1,\dots,n\}} \hat{y}_i + \mathbf{g}_i^T(\x - \x_i).
\end{equation}The characterization~\citep{Ku08} and consistency~\citep{SeSe11,LiGl11} of the least squares problem have only recently been studied. The LSE quickly becomes impractical due to its size:  Equation (\ref{eq:lse}) has $n(n-1)$ constraints. This results in a computational complexity of $\bigO((p+1)^3 n^3)$, which becomes impractical after one to two thousand observations.

% Multivariate: 
% - least squares: Hi54 (original), Ho79 (formulation), Ku08 (formulation and sufficiency), SeSe11 (consistency)
% 	~ QP with n^2 constraints (cone constraints)
% - entropy methods: AlBeHa07
% 	~ LP with n^2 constraints
% - constrained kernel methods: HePa09
% 	~ sequential quadratic programming, slow
% - linear constraints (over triangularization of X): CaLaMa01
% 	~ affine approximation over triangularization of X; only suitable for R^2
% - psd Hessian: AgMo08, AgMo09
% 	~ convergent but slow
% - projection of convex hull of kernel methods: AgFoMo11
% 	~ issues of smoothing methods in multiple dimensions (does *not* work for high dimensions), but consistent
% - iterative fitting: MaBo09
% 	~ non-convergent, non-consistent and choice of k is unexplored---but fast

While the LSE is widely studied across all fields, the remaining literature on multivariate convex regression is sparser and more dispersed than the univariate literature. One approach is to place a positive semi-definite restriction on the Hessian of the estimator. In the economics literature, \citet{HePa09} used kernel smoothing with a restricted Hessian and found a solution with sequential quadratic programming. In electrical engineering, \citet{RoChCh07}, and in a variational setting, \citet{AgMo08} and \citet{AgMo09}, used semi-definite programming to search the space of functions with positive semi-definite local Hessians. Although consistent in some cases~\citep{AgMo08,AgMo09}, Hessian methods are computationally intensive and can be poorly conditioned in boundary regions. In another approach, \citet{AlBeHa07} proposed a method based on reformulating the maximum likelihood problem as one minimizing entropic distance, which can be solved as a linear program. However, like the original maximum likelihood problem, the transformed problem still has $n^2$ constraints and does not scale to more than a few thousand observations. 

Recently, multivariate convex regression methods have been proposed with a more traditional statistics approach. \citet{AgFoMo11} proposed a two step smoothing and fitting process. First, the data were smoothed and functional estimates were generated over an $\epsilon$-net over the domain. Then the convex hull of the smoothed estimate was used as a convex estimator. Again, although this method is consistent, it is sensitive to the choice of smoothing parameter and does not scale to more than a few dimensions. \citet{HaDu11c} proposed a Bayesian model that placed a prior over the set of all piecewise linear models. They were able to show adaptive rates of convergence, but the inference algorithm did not scale to more than a few thousand observations.

In a more computational approach, \citet{MaBo09} use an iterative fitting scheme. In this method, the data were divided into $K$ random subsets and a linear model was fit within each subset; a convex function was generated by taking the maximum over these hyperplanes. The hyperplanes induce a new partition, which is then used to refit the function. This sequence was repeated until convergence. Despite relatively strong empirical performance, this method is sensitive to the initial partition and the choice of $K$. Moreover, it is not consistent and there are cases when the algorithm does not even converge.

%=====================================================
%  Algorithm
%=====================================================

\section{The CAP Algorithm}\label{sec:CAP}

% This section:
% - how this algorithm is a sparse version of LSE (less computation AND less overfitting!)
% 	+ constructed with adaptive partitioning methods
% 	+ why not just use a tree with linear leaves and take the maximum?
% 		~ results aren't that good (empirically)
%		~ *use* shape constraints: piecewise linear model induces a partition, so partition that matches linear model is at least a local optimum
%		~ simply refitting may not be stable (MaBo09, Br91(3?))
% - Subsection: algorithm
% - Subsection: tunable parameters
% - Subsection: implementation and speedups

As seen in much of the literature, a natural way to model a convex function $f_0$ is through the maximum of a set of hyperplanes. One example of this method is the least squares estimator, which fits every observation with its own hyperplane. This is computationally expensive and can result in overfitting, as shown in Figure \ref{fig:CAPvLSE}. Instead, we wish to model $f_0$ through only $K$ hyperplanes. We do this by partitioning the covariate space and approximating the gradients within each region by hyperplanes generated by the least squares estimator. The covariate space partition and $K$ are chosen through adaptive partitioning.

\begin{figure}
\begin{center}
\includegraphics[width=3.22in, viewport = 110 260 510 550]{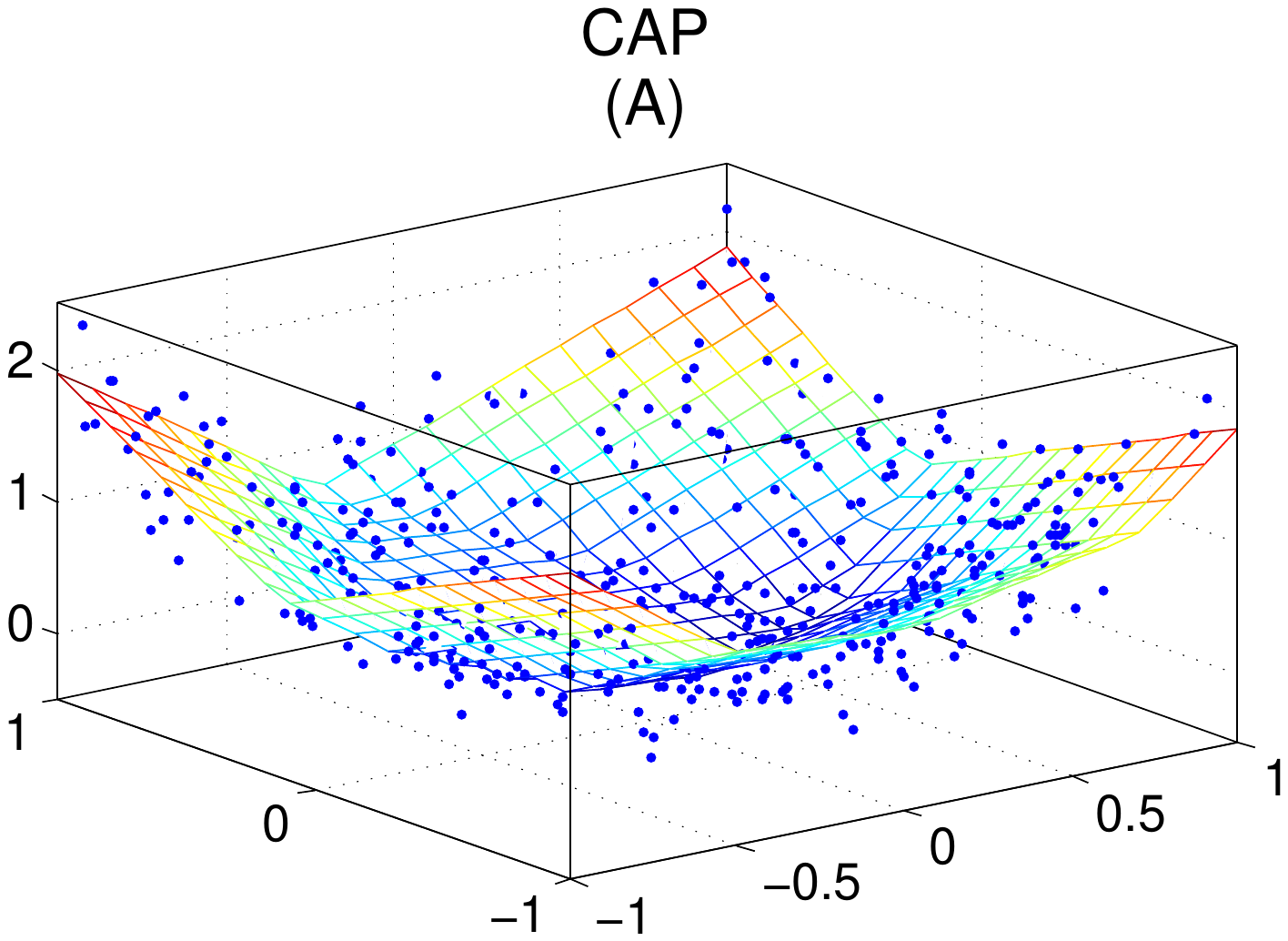}
\includegraphics[width=3.22in, viewport = 110 260 510 550]{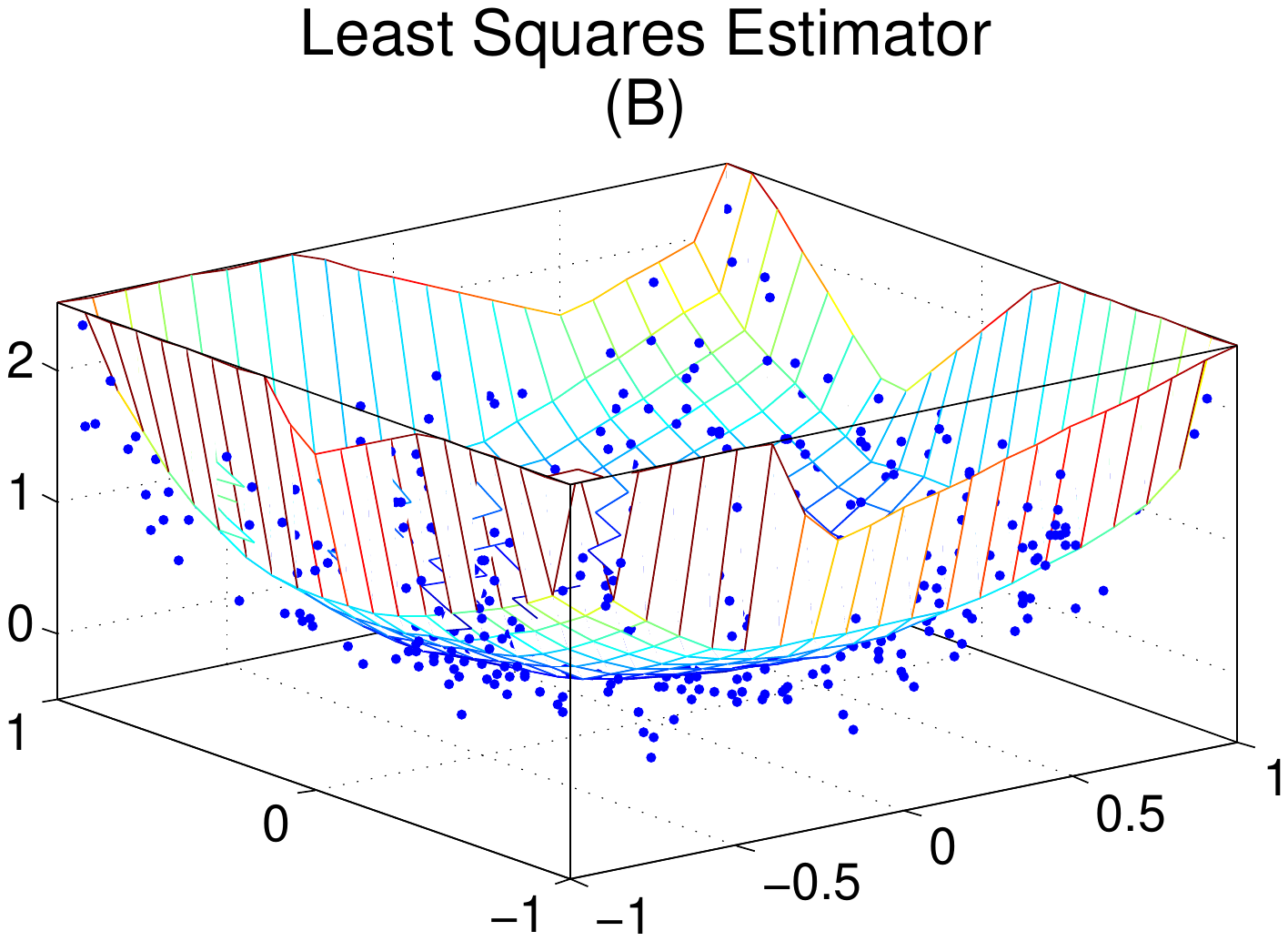}
\end{center}
\caption{(A) The CAP estimator, and (B) the LSE fit to 500 observations drawn from $y = x_1^2 + x_2^2 + \epsilon,$ where $\epsilon \sim N(0,0.25^2)$. The covariates were drawn from a 2 dimensional uniform distribution, $\mathrm{Unif}[-1,1]^2$. The LSE was truncated at predicted values of 2.5 for display, although some predicted values reached as high as 4,800 on $[-1,1]^2$.}
\label{fig:CAPvLSE}
\end{figure}

Given a partition $\{A_1,\dots,A_K\}$ of $\mathcal{X}$, an estimate of the gradient for each subset can be created by taking the least squares linear estimate based on all of the observations within that region,
\begin{equation}\notag
(\alpha_k, \beta_k) = \arg \min_{\alpha,\beta} \sum_{i \, : \, \x_i \in A_k} \left(y_i - \alpha - \beta^T \x_i\right)^2.
\end{equation} A convex function $\hat{f}$ can be created by taking the maximum over $(\alpha_k,\beta_k)_{k=1}^K$,
\begin{equation}\notag
\hat{f}(\x) = \max_{k \in \{ 1,\dots,K\}} \alpha_k + \beta_k^T \x.
\end{equation}

% How are these partitions usually created
Models adaptive partitioning models with linear leaves have been proposed before; see \citet{ChHuLo94,ChLoLo95,AlGr96,No96,DoGe02,GyKoKr02} and \citet{PoDu05} for examples. In most of these cases, the partition is created by adaptively refining an existing partition by dyadic splitting of one subset. The split is chosen in a way that minimizes local error within the subset. There are two problems with these partitioning methods that arise when a piecewise linear summation function,
\begin{equation}\notag
f^*(\x) = \sum_{k = 1}^K \left(\alpha_k + \beta_k^T\x\right) \1_{\{ \x \in A_k\}},
\end{equation} is changed into a piecewise linear maximization function, like $\hat{f}$. First, a split that minimizes local error does not necessarily minimize global error for $\hat{f}$. This is fairly easy to remedy by considering splits based on minimizing global error. The second problem is more difficult: the gradients often act in areas over which they were not estimated.

% Problems with max; refitting as Gauss-Newton

A piecewise linear maximization function, $\hat{f}$, generates a new partition, $\{A^{\prime}_1,\dots,A_K^{\prime}\}$, by
\begin{equation}\notag
A_k^{\prime} = \left\{ \x \in \mathcal{X} \, : \, \alpha_k + \beta_k^T \x > \alpha_j + \beta_j^T \x, \forall \ j \neq k\right\}.
\end{equation} The partition $\{A_1,\dots,A_K\}$ is not necessarily the same as $\{A^{\prime}_1,\dots,A^{\prime}_K\}$. We can use this new partition to refit the hyperplanes and produce a significantly better estimate. Refitting hyperplanes in this manner can be viewed as a Gauss-Newton method for the non-linear least squares problem~\citep{MaBo09},
\begin{equation}\notag
\mathrm{minimize} \ \sum_{i=1}^n \left(y_i - \max_{k\in \{1,\dots,K\}} \left(\alpha_k + \beta_k^T \x_i\right) \right)^2.
\end{equation} Similar methods for refitting hyperplanes have been proposed in \citet{Br93} and \citet{MaBo09}. However, repeated refitting may not converge to a stationary partition and is sensitive to the initial partition.

Convex Adaptive Partitioning (CAP) uses adaptive partitioning with linear leaves to fit a convex function that is defined as the maximum over the set of leaves. The adaptive partitioning itself differs from previous methods in order to fit piecewise linear maximization functions. Partitions are refined in two steps. First, candidate splits are generated through dyadic splits of existing partitions. These are evaluated and the one that minimizes global error is greedily selected. Second, the new partition is then refit. Although simple, these rules, and refitting in particular, produce large gains over naive adaptive partitioning methods; empirical results are discussed in Section \ref{sec:numbers}.

Most other adaptive partitioning methods use backfitting or pruning to select the tree or partition size. Due to the construction of the CAP estimator, we cannot locally prune and so instead we rely on model selection criteria. We derive a generalized cross-validation method for this setting that is used to select $K$. This is discussed in Section \ref{sec:gcv}.

% Other benefits
Since CAP shares many feature with existing adaptive partitioning algorithms, we are able to use many adaptive partitioning results to study CAP practically and theoretically. This is in stark contrast to the least squares estimator. Despite being introduced by \citet{Hi54}, it has not been implemented in many practical settings~\citep{Li10} and has only very recently been shown to be consistent~\citep{SeSe11,LiGl11}. 

\subsection{The Algorithm}
We now introduce some notation required for Convex Adaptive Partitioning. When presented with data, a partition can be defined over the covariate space (denoted by $\{A_1,\dots,A_K\}$, with $A_k \subseteq \mathcal{X}$) or over the observation space (denoted by $\{C_1,\dots,C_K\}$, with $C_k \subseteq \{1,\dots,n\}$). The observation partition is defined from the covariate partition,
\begin{equation}\notag
C_k = \left\{ i \, : \, \x_i \in A_k\right\}, \quad k = 1,\dots, K.
\end{equation}
CAP proposes and searches over a set of models, $M_1,\dots,M_K$. A model $M_k$ is defined by: 1) the covariate partition $\{A_1,\dots,A_K\}$, 2) the corresponding observation partition, $\{C_1,\dots,C_k\}$, and 3) the hyperplanes $(\alpha_j,\beta_j)_{j=1}^k$ fit to those partitions.

The CAP algorithm progressively refines the partition until each subset cannot be split without one subset having fewer than a minimal number of observations, $n_min$, where
\begin{equation}\notag
n_{min} = \min \left\{ \frac{n}{D \log (n)}, 2(d+1)\right\}.
\end{equation} Here $D$ is a log scaling factor, which acts to change the base of the log operator. This minimal number is chosen so that 1) there are enough observations to accurately fit a hyperplane, and 2) there is a lower bound on the growth rate for the number of observations in each subset---and an upper bound on the number of subsets. This is used to show consistency.

We briefly outline the CAP algorithm below. Details are given in the following subsections.

\noindent \paragraph{Convex Adaptive Partitioning (CAP)}
\begin{enumerate}
	\item {\bf Initialize.} Set $K=1$; place all observations into a single observation subset, $C_1 = \{1,\dots,n\}$; $A_1 = \mathcal{X}$; this defines model $M_1$.
	%Place all observations into a single observation subset, $C_1 = \{1,\dots,n\}$, and the entire domain into a single covariate subset, $A_1 = \mathcal{X}$. Set number of subsets, $K$, to 1. The model $M_1$ is defined by the observation partition, $C_1$, the covariate partition, $A_1$, and the hyperplane $(\alpha_1,\beta_1)$.
	\item{\bf Split.}\label{item:split} Refine partition by splitting a subset.
	\begin{enumerate}[a.]
		\item{\it Generate candidate splits.} Generate candidate model $\hat{M}_{kj\ell}$ by 1) fixing a subset $k$, 2) fixing a dimension $j$, 3) dyadically dividing the data in subset $k$ and dimensions $j$ according to knot $a_{\ell}$. This is done for $L$ knots, all $p$ dimensions and $K$ subsets.
	\item {\it Select split.} Choose the model $M_{K+1}$ from the candidates that minimizes global mean squared error on the training set and satisfies $\min_k |C_k| \geq n_{min}$. Set $K = K+1$.
	\end{enumerate}
	\item {\bf Refit.} Use the partition induced by the hyperplanes to generate model $\hat{M}_k$. Set $M_K = \hat{M}_K$ if for every subset $\hat{C}_k$ in $\hat{M}_k$, $|\hat{C}_k| \geq n_{min}$. % and $\max_k \mathrm{diam}(\hat{A}_k)\leq \max_k \mathrm{diam}(A_k)$. %(Additional conditions based on geometric size of $\hat{C}_k$ may be added to guarantee consistency.)
	\item {\bf Stopping conditions.} If for every subset $C_k$ in $M_k$, $|C_k| < 2n_{min}$, stop fitting and proceed to step \ref{item:modelSize}. Otherwise, go to step \ref{item:split}. 
	\item\label{item:modelSize} {\bf Select model size.} Each model $M_k$ creates an estimator, $$\hat{f}_{nk}(\x) = \max_{j\in \{1,\dots,k\}} \alpha_j + \beta_j^T \x.$$
	Use generalized cross-validation on the estimators to select final model $M^*$ from $\{M_k\}_{k=1}^K$. 
\end{enumerate}

%\subsection{Initialization}
%Set $K$, the number of partitions, equal to 1 and place all observations in the same subset, $C_1$. Let $A_1 \subseteq \mathcal{X}$ be the subset of $\mathcal{X}$ that defines $C_1$. In this case, $A_1 = \mathcal{X}$. Fit a least squares hyperplane $(\alpha_1,\beta_1)$ to the observations in $C_1$. The model $M_1$ is defined by the observation partition, $C_1$, the covariate partition, $A_1$, and the hyperplane $(\alpha_1,\beta_1)$.

\subsection{Splitting Rules}To split, we create a collection of candidate models by dyadically splitting a single subset. Since the best way to do this is not apparent, we create models for every subset and search along every cardinal direction by splitting the data along that direction. We create model $\hat{M}_{kj\ell}$ by 1) fixing subset $k \in \{1,\dots,K\}$, and 2) fixing dimension $j\in \{1,\dots,p\}$. Let $x_{min}^{kj}$ be the minimum value and $x_{max}^{kj}$ be the maximum value of the covariates in this subset and dimension,
\begin{align}\notag
x_{min}^{kj} & = \min\{ x_{ij} \, : \, i \in C_k\}, & x_{max}^{kj} & = \max\{ x_{ij} \, : \, i \in C_k\}.
\end{align}Let $0 <  a_1 <\dots < a_L < 1$ be a set of evenly spaced knots that represent the proportion between $x_{min}^{kj}$ and $x_{max}^{kj}$. %For example, if $L = 3$, then $a_1 = 1/4$, $a_2 = 1/2$ and $a_3 = 1/3$. 

Use the weighted average $b_{kj\ell} = a_{\ell} x_{min}^{kj} + (1- a_{\ell})x_{max}^{kj}$ to split $C_k$ and $A_k$ in dimension $j$. Set
\begin{align}\notag
\hat{C}_k & = \{i : i \, \in C_k, \, x_{ij} \leq b_{kj\ell}\}, & \hat{C}_{K+1} & = \{i : i \, \in C_k, \, x_{ij} > b_{kj\ell}\},\\\notag
\hat{A}_k & = \{\x : \x \, \in A_k, \, x_{j} \leq b_{kj\ell}\}, & \hat{A}_{K+1} & = \{\x : \x \, \in A_k, \, x_{j} > b_{kj\ell}\}.
\end{align}These define new subset and covariate partitions, $\hat{C}_{1:K+1}$ and $\hat{A}_{1:K+1}$ where $\hat{C}_{k^{\prime}} = C_{k^{\prime}}$ and $\hat{C}_{k^{\prime}} = C_{k^{\prime}}$ for $k^{\prime}\neq k$. Fit hyperplanes $(\hat{\alpha}_k,\hat{\beta}_k)_{k=1}^{K+1}$ in each of the subsets. The triplet of observation partition $\hat{C}_{1:K+1}$, covariate partition, $\hat{A}_{1:K+1}$, and set of hyperplanes $(\hat{\alpha}_k,\hat{\beta}_k)_{k=1}^{K+1}$ defines the model $\hat{M}_{kj\ell}$. This is done for $k = 1,\dots,K$, $j = 1,\dots,p$ and $\ell = 1,\dots,L$. After all models are generated, set $K = K+1$.

We note that any models where $\min_k |\hat{C}_k| < n_{min}$ are discarded. If all models are discarded in one subset/dimension pair, we produce a model by splitting on the subset median in that dimension. % or $$\max_k \mathrm{diam}(\hat{A}_k) \geq \max_k \mathrm{diam}(A_k)$$ are discarded. If all models are discarded in one subset/dimension pair, we produce a model by splitting on the subset median in that dimension and/or splitting along the hyperplane normal to the line that defines $\mathrm{diam}(A_k)$.

%Other more elegant splitting rules have been proposed for fitting tree-based models, such as those in \citet{BrFrOl84} and \citet{ChHuLo94}. In all of these methods, a local gain in the mean squared error is equivalent to a global gain the mean squared error. By taking the maximum over the hyperplanes, however, CAP fits a model where error gains can only be evaluated globally. In this situation, there is no easy way to find the best split, so we approximate it by greedily selecting from a group of models. Our search method closely resembles the search method for Multivariate Adaptive Regression Splines (MARS)~\citep{Fr91}. 

\subsection{Split Selection}
We select the model $\hat{M}_{kj\ell}$ that gives the smallest {\it global} error. Let $(\alpha_{i}^{kj\ell},\beta_i^{kj\ell})_{i=1}^K$ be the hyperplanes associated with $\hat{M}_{kj\ell}$ and let
\begin{align}\notag
\hat{f}_{nk}^{j\ell} (\x) & = \max_{i \in \{1,\dots,K\}} \alpha_{i}^{kj\ell} + {\beta_i^{kj\ell}}^T \x
\end{align} be its estimator. We set the model $M_K$ to be the one that minimizes global mean squared error,
\begin{align}\notag
M_K = \left\{ \hat{M}_{kj\ell} \, : \, (k,j,\ell) = \arg \min_{k,j,\ell} \frac{1}{n} \sum_{i=1}^n \left( y_i - \hat{f}_{nk}^{j\ell}(\x_i)\right)^2\right\}.
\end{align}Set $\hat{f}_{nK}$ to be the minimal estimator.

\subsection{Refitting}
We refit by using the partition induced by the hyperplanes.  Let $(\alpha_{1:K},\beta_{1:K})$ be the hyperplanes associated with $M_K$. Refit the partitions by $$\hat{C}_k = \{\x_i:\alpha_k + \beta_k^T\x_i \geq \alpha_j + \beta_j^T\x_i, j \neq k\}$$ for $k = 1,\dots,K$. The covariate partition, $\hat{A}_{1:K}$ is defined in a similar manner. Fit hyperplanes in each of those subsets. Let $\hat{M}_K$ be the model generated by the partition $\hat{C}_1,\dots,\hat{C}_K$. Set $M_K = \hat{M}_K$ if $|\hat{C}_k| \geq n_{min}$ for all $k$.
%\begin{enumerate}
%	\item $|\hat{C}_k| \geq n / (D \log(n))$ for all $k$, and 
%	\item $\max_k \mathrm{diam}(\hat{A}_k) \geq \max_k \mathrm{diam}(A_k)$.
%\end{enumerate}

%\subsection{Stopping Conditions}
%f $|C_k| < 2n/(D \log(n))$ for every $k$, stop fitting and choose CAP the model by generalized cross-validation.

\subsection{Tunable Parameters}
CAP has two tunable parameters, $L$ and $D$. $L$ specifies the number of knots used when generating candidate models for a split. Its value is tied to the smoothness of $f_0$ and after a certain value, usually 5 to 10 for most functions, higher values of $L$ offer little fitting gain.

The parameter $D$ is used to specify a minimum subset size, $|C_k| \geq n/(D \log(n))$. Here $D$ transforms the base of the logarithm from $e$ into $\exp(1/D)$. We have found that $D = 3$ (implying base $\approx 1.4$) is a good choice for most problems.

Increases in either of these parameters increase the computational time. Sensitivity to these parameters, both in terms of predictive error and computational time, is empirically examined in Section \ref{sec:sensitivity}.

%=====================================================
%  Theory
%=====================================================

\section{Theoretical Properties of CAP}\label{sec:theory} 

In this section, we give the computational complexity for the CAP algorithm and conditions for consistency. Since CAP is similar to existing adaptive partitioning methods, we can leverage existing results to show consistency. 

\subsection{Computational Complexity}
Computational complexity describes the number of bit operations a computer must do to perform a routine, such as CAP. It is useful to determine small sample runtimes and how well routines will scale to larger problems. The computational complexity of the least squares estimator is unworkably high at $\bigO((p+1)^3 n^3)$ flops to solve a problem with $n$ observations in $p$ dimensions~\citep{MoAd89}. The worst case computational complexity of CAP is much lower, at $\bigO( p(p+1)^2 n \log(n) \log( \log(n)))$ flops when implemented as in Section \ref{sec:CAP}. The most demanding part of the CAP algorithm is the linear regression; each one has complexity $\bigO((p+1)^2n)$. For iteration $k$ of the algorithm, $L p k$ linear regressions are fit. This is done for $k = 1,\dots, K$, where $K$ is bounded by $D\log(n)$. Putting this together we obtain the above complexity.

\begin{figure}
\begin{center}
\includegraphics[width=6.5in, viewport = 10 260 595 535]{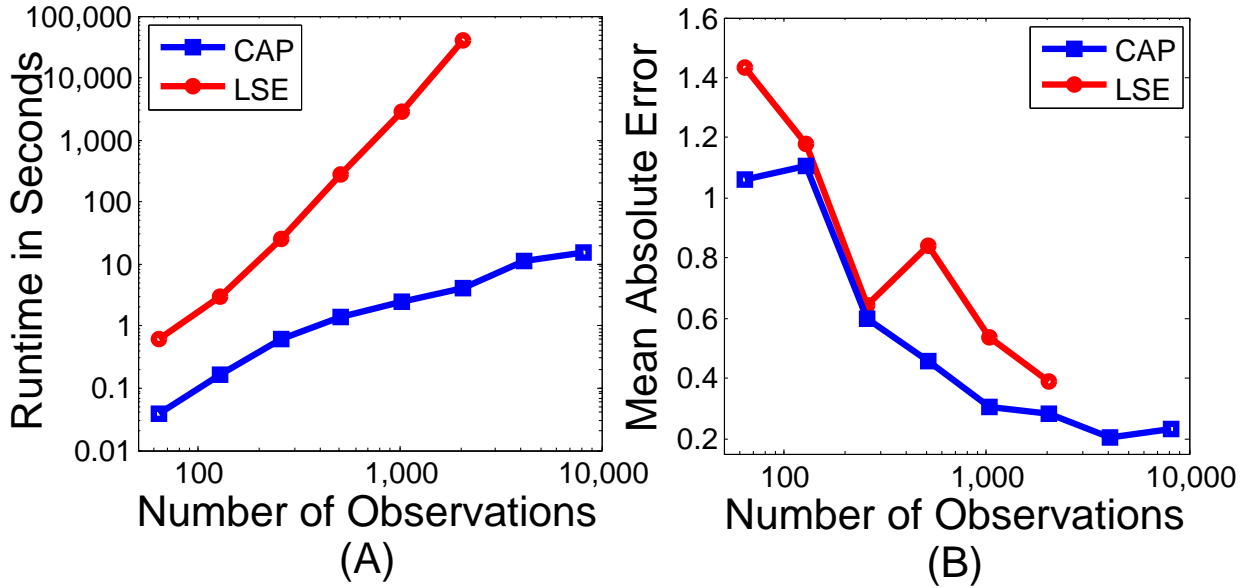}
\end{center}
\caption{(A) Number of observations $n$ (log scale) vs. runtime in minutes (log scale), and (B) number of observations $n$ (log scale) vs. mean absolute error (linear scale) for the least squares estimator (LSE), and CAP. Here $\x \in \R^5$ and $y = \left(x_1+.5x_2+x_3\right)^2 - x_4+.25x_5^2 + \epsilon,$
where $\epsilon \sim N(0,1)$. The covariates are drawn from a 5 dimensional standard Gaussian distribution, $N_5(0,I)$.}
\label{fig:complexity}
\end{figure}

To demonstrate how much these factors matter in practice, we empirically compare CAP, Fast CAP and LSE on a small problem, $y = \left(x_1+.5x_2+x_3\right)^2 - x_4+.25x_5^2 + \epsilon,$ where $\x \in \R^{5}$, $\x \sim N_{5}(0,I)$ and $\epsilon \sim N(0,1)$. The runtimes and mean absolute errors of each method are shown in Figure \ref{fig:complexity}. %The LSE becomes computationally infeasible around $n = 1,000$, while CAP and Fast CAP run in under 30 seconds for problems a order of magnitude larger. 

% LSE would take an estimate 46 days to run on the largest problem size above!

\subsection{Consistency}\label{sec:consistency}

% - game plan:
% 	+ consistency result is similar to that for other adaptive partitioning models
% 	+ will use proof like ChHuLo94, but need to extend to case with polyhedral partitions that have increasing numbers of faces
% 	+ extend consistency from \1 function to \max function
% - define needed terms: d_{nk}, \hat{f}_n, f^*_n, K --> K_n
% - give assumptions
% - explain assumptions
% - give theorem for \1 function
% - give theorem for CAP

We now show consistency for the CAP algorithm. Consistency is shown in a similar manner to consistency for other adaptive partitioning models, like CART~\citep{BrFrOl84}, treed linear models~\citep{ChHuLo94} and other variants~\citep{No96,GyKoKr02}. We take a two-step approach, first showing consistency for the mean function and first derivatives of a more traditional treed linear model based on CAP under the $\ell_{\infty}$ metric and then we use that to show consistency for the CAP estimator itself. 

Letting $M_n^*$ be the model for the CAP estimate after $n$ observations, define the discontinuous piecewise linear estimate based on $M_n^*$,
\begin{align}\notag
f^*_n(\x) & = \sum_{k=1}^{K_n} \left(\alpha_k + \beta_k^T \x\right) \1_{\{\x \in A_k\}},
\end{align}where $K_n$ is the partition size, $A_1,\dots,A_{K_n}$ are the covariate partitions and $(\alpha_k,\beta_k)_{k=1}^{K_n}$ are the hyperplanes associated with $M_n^*$. Likewise, let $\hat{f}_n(\x)$ be the CAP estimator based on $M_n^*$,
\begin{align}\notag
\hat{f}_n(\x) & = \max_{k \in \{ 1,\dots,K_n\}} \alpha_k + \beta_k^T \x.
\end{align}Each subset $A_k$ has an associated diameter, $d_{nk}$, where $$d_{nk} = \sup_{\x_1,\x_2 \in A_k} ||x_1 - x_2||_2.$$ Define the empirical covariate mean for subset $k$ as $\bar{x}_k = \frac{1}{|C_k|} \sum_{i \in C_k } \x_i.$ For $\x_i \in A_k,$ define
\begin{align}\notag
\Gamma_i & =   \left[\begin{array}{c}
    [1,\dots, 1] \\ 
     d_{nk}^{-1} \left( \x_i - \bar{\x}_k\right)  
\end{array}\right], & G_k & = \sum_{i \in C_k} \Gamma_i \Gamma_i^T.
\end{align}Note that $(\alpha_k, \beta_k) = G_k^{-1} \sum_{i \in C_k} \Gamma_i y_i$ whenever $G_k$ is nonsingular.

%Consistency is shown in the same manner as consistency for piecewise-polynomial functions in \citet{ChHuLo94}. We first show $\ell_{\infty}$ consistency in probability of both the mean function and the derivatives for the piecewise linear models, $f^{*}_n$. We modify the methods of \citet{ChHuLo94} and \citet{BrFrOl84} to allow adaptively selected partitions defined by polynomial regions where the number of faces grows like $\log(n)$ at most. Consistency results for these types of partitions with respect to the mean function under the $\ell_2$ metric have been shown by \citet{No96} and \citet{GyKoKr02}. We use $\ell_{\infty}$ consistency results along with Lipschitz continuity to show convergence of $\hat{f}_n$ to $f_0$ with respect to the $\ell_{\infty}$ metric in probability.

%Let $\hat{f}_n^*$ be the terminal piecewise linear function after $n$ observations and $\hat{f}_n$ the convex function generated by $\hat{f}_n^*$. Let $K_n$ be the number of hyperplanes and let $C_1, \dots, C_{K_n}$ and $A_1, \dots, A_{K_n}$ be the subsets that generated $\hat{f}_n^*$. 

Let $\x_1,\dots,\x_n$ be i.i.d. random variables. We make the following assumptions:
\begin{enumerate}
	\item[{\bf A1.}] $\mathcal{X}$ is compact and $f_0$ is Lipschitz continuous and continuously differentiable on $\mathcal{X}$ with Lipschitz parameter $\zeta$.
	\item[{\bf A2.}] There is an $a > 0$ such that $\E\left[ e^{a | Y - f_0(x)| }\g \X = \x\right]$ is bounded on $\mathcal{X}$.
	%\item[{\bf A2.}] The covariates $\x_1,\dots,\x_n \sim_{iid} \mu$, where $\mu$ is a non-atomic measure that is absolutely continuous with respect to the Legesgue measure.
	\item[{\bf A3.}] Let $\lambda_k$ be the smallest eigenvalue of $|C_k|^{-1} G_k$ and $\lambda_n = \min_k \lambda_k$. Then $\lambda_n$ remains bounded away from 0 in probability as $n \rightarrow \infty$.
	\item[{\bf A4.}] The diameter of the partition $\max_k d_{nk}^{-1} \rightarrow 0$ in probability as $n \rightarrow \infty$.
\end{enumerate}

Assumptions {\bf A1.} and {\bf A2.} place regularity conditions on $f_0$ and the noise distribution, respectively. Assumption {\bf A3.} is a regularity condition on the covariate distribution to ensure the uniqueness of the linear estimates. Assumption {\bf A4.} is a condition that can be included in the algorithm and checked along with the subset cardinality, $|C_k|$. If $\mathcal{X}$ is given, it can be computed directly, otherwise it can be approximated using $\{\x_i \, : \, i\in C_k\}$. In some cases, such as when $f_0$ is strongly convex, {\bf A4.} will be satisfied without enforcement due to problem structure.

To show consistency of $\hat{f}_n$ under the $\ell_{\infty}$ metric, we first show consistency of $f^*_n$ and its derivatives under the $\ell_{\infty}$ metric in Theorem \ref{thm:fK}. This is very close to Theorem 1 of \citet{ChHuLo94} for treed linear models, although we need to modify it to allow partitions with an arbitrarily large number of faces.

\begin{thm}\label{thm:fK}
Suppose that assumptions {\bf A1.} through {\bf A4.} hold. Then, 
\begin{align}\notag
\max_{k =1,\dots,K_n} \sup_{\x \in A_k} \left| \alpha_k + \beta_k^T \x  - f_0(\x) \right| & \rightarrow 0, &\max_{k =1,\dots,K_n} \sup_{\x \in A_k} \left|\left| \beta_k  - \nabla f_0(\x) \right|\right|_{\infty} & \rightarrow 0
\end{align}in probability as $n\rightarrow \infty$.
\end{thm}
The CAP algorithm is similar to the SUPPORT algorithm of \citet{ChHuLo94}, except the refitting step of CAP allows partition subsets to be polyhedra with up to $D\log(n)$ faces. Theorem \ref{thm:fK} is analogous to Theorem 1 of \citet{ChHuLo94}; to prove our theorem, we modify parts of the proof in \citet{ChHuLo94} that rely on a fixed number of polyhedral faces. As such, we first need to modify Lemma 12.27 of \citet{BrFrOl84}.

\begin{lem}[Modification of Lemma 12.27 of \citet{BrFrOl84}]\label{lem:12.27}
Suppose that {\bf A2.} holds and that there exists a $k_n \rightarrow \infty$ where $k_n/ \log(n)\rightarrow \infty$. Then, for every compact set $B$ in $\mathcal{X}$ and every $\epsilon > 0 $ and $c > 0$,
\begin{align}\notag
\lim_{n \rightarrow \infty} n^c \, \p \left( \left| \frac{1}{|C_k|} \sum_{i \in C_k} y_i - \frac{1}{|C_k|} \sum_{i \in C_k} f_0(\x_i)\right| > \epsilon \mathrm{ \ for \  } k \mathrm{ \ such \ that \ } A_k \subset B, \  |C_k| \geq k_n D \log (n)\right).
\end{align}
\end{lem}
\begin{proof}
To prove this Lemma, we only need to lift the restriction on the number of faces of the polyhedron $A_k$ from being bounded by a fixed $M$ to $D\log (n)$. First, we note that $|C_k| \geq n_{min}$ implies that $K_n \leq D \log(n).$ Following the proof in \citet{BrFrOl84}, we note that
\begin{align}\notag
\p & \left( \left| \frac{1}{|C_k|} \sum_{i \in C_k} y_i - \frac{1}{|C_k|} \sum_{i \in C_k} f_0(\x_i)\right| > \epsilon \mathrm{ \ for \  } k \mathrm{ \ such \ that \ } A_k \subset B, \  |C_k| \geq k_n D \log (n)\right)\\\notag
& \leq 2^{1 + (D\log(n))(p+2)} n^{(D\log(n))(p+2) - \epsilon^2 k_n / 2C},
\end{align} for a fixed constant $C$ depending on assumption 6. Since $k_n/\log(n) \rightarrow \infty$, the conclusion holds.
\end{proof}

With Lemma \ref{lem:12.27}, the proof of Theorem \ref{thm:fK} follows directly from the arguments of \citet{ChHuLo94}.

Using the results from Theorem \ref{thm:fK}, extension to consistency for $\hat{f}_n$ under the $\ell_{\infty}$ metric is fairly simple; this is given in Theorem \ref{thm:consistency}.

\begin{thm}\label{thm:consistency}
Suppose that assumptions {\bf A1.} through {\bf A4.} hold. Then, 
\begin{align}\notag
\sup_{\x \in \mathcal{X}} \left| \hat{f}_n(\x)  - f_0(\x) \right| & \rightarrow 0
\end{align}in probability as $n\rightarrow \infty$.
\end{thm}
\begin{proof}Fix $\epsilon > 0$; let $d_{\mathcal{X}}$ be the diameter of $\mathcal{X}$. Choose $n$ such that
\begin{align}\notag
\p \left\{ \max_{k =1,\dots,K_n} \sup_{\x \in A_k} \left| \alpha_k + \beta_k^T \x  - f_0(\x) \right|  > \frac{\epsilon}{\zeta d_{\mathcal{X}}} \right\}& < \epsilon/2, \\\notag
\p \left\{ \max_{k =1,\dots,K_n} \sup_{\x \in A_k} \left|\left| \beta_k  - \nabla f_0(\x) \right|\right|_{\infty} >  \frac{\epsilon}{\zeta d_{\mathcal{X}}} \right\}& < \epsilon/2
\end{align}
Fix a $\delta$ net over $\mathcal{X}$ such that at least one point of the net sits in $A_k$ for each $k = 1, \dots, K$. Let $n_{\delta}$ be the number of points in the net and let $\x_{i}^{\delta}$ be a point. Then,
\begin{align}\notag
\p\left\{ \sup_{x\in \mathcal{X}} \left| \hat{f}_n(\x) - f_0(\x) \right| > \epsilon \right\} & = \p\left\{ \sup_{x\in \mathcal{X}} \left| \max_{k=1,\dots,K_n}\alpha_k + \beta_k^T \x - f_0(\x) \right| > \epsilon\right\},\\\notag
& \leq  \p\left\{ \max_{i = 1,\dots, n_{\delta}} \left| \max_{k=1,\dots,K_n}\alpha_k + \beta_k^T \x_i^{\delta} - f_0(\x_i^{\delta}) \right| > \frac{\epsilon}{\zeta}\right\},\\\notag
& \leq \p\left\{ \max_{i = 1,\dots, n_{\delta}} \left| \sum_{k=1}^{K_n} \left(\alpha_k + \beta_k^T \x_i^{\delta}\right)\1_{\{x_i^{\delta} \in A_k\}} - f_0(\x_i^{\delta}) \right| > \frac{\epsilon}{\zeta d_{\mathcal{X}}}\right\},\\\notag
& < \epsilon.\end{align}\end{proof}

\section{Model Selection and Modifications}\label{sec:implementation}

The terminal model produced by the CAP algorithm often overfits the data and is computationally more intensive than necessary. In this section, we derive a generalized cross-validation method to select the best model from all of those produced by CAP, $M_1,\dots,M_K$. We then propose an approximate algorithm, Fast CAP, that requires substantially less computation than the original algorithm.

% Discuss:
%  - tunable parameters
%  - model selection
%  - speedups

%\subsection{Implementation}
%The shrinking partition diameter condition is necessary for the proof of consistency, but computing $\mathrm{diam}(A_k)$ directly is not possible unless $\mathcal{X}$ is known. This is not the case most of the time. In situations with $n$ large and a well-sampled space, $\mathrm{diam}(A_k)$ may be well approximated by using the observations in $C_k$. However, the number of samples required may be very large so we omit this condition when implementing CAP.

%We believe that including it may actually hurt performance and slow down our empirical rates of convergence by forcing the partition to be small in {\it all} directions, even when $f_0$ may only map from a lower dimensional subspace of $\mathcal{X}$. See Section \ref{sec:rates} for a discussion of empirical convergence rates.  

\subsection{Generalized Cross-Validation for CAP}\label{sec:gcv}
Cross-validation is a method to assess the predictive performance of statistical models and is routinely used to choose tunable parameters. In this case, we would like to choose the cardinality of the partition, $K$. As a fast approximation to leave-one-out cross-validation, we use generalized cross-validation~\citep{GoHeWa79,Fr91}. In a linear regression setting,
\begin{align}\label{eq:linearGCV}
\frac{1}{n} \sum_{i=1}^n \left(y_i - \hat{f}_{-i}(\x_i)\right)^2 & = \frac{1}{n} \sum_{i=1}^n \left(\frac{y_i - \hat{f}_n(\x_i)}{1 - H_{ii}}\right)^2
 \approx \frac{1}{n}\sum_{i=1}^n \left(\frac{y_i - \hat{f}_n(\x_i)}{1 - \frac{\nu}{n}}\right)^2,
\end{align}where $H_{ii}$ is the $i^{th}$ diagonal element of the hat matrix, $\X (\X^T \X)^{-1} \X^T$, $\hat{f}_{-i}$ is the estimator conditioned on all of the data minus element $i$, and $\nu$ is the degrees of freedom.

A given model $M_K$ is generated by a collection of linear models. A similar type approximation to leave-one-out cross-validation can be used to select the model size. The model $M_K$ is defined by $C_1,\dots,C_K$, the partition, and the hyperplanes $(\alpha_k,\beta_k)_{k=1}^K$, which were generated by the partition. Let $(\alpha_k^{(-i)},\beta_k^{(-i)})_{k=1}^K$ be the collection of hyperplanes generated when observation $i$ is removed; notice that if $i \in C_k$, only $(\alpha_k,\beta_k)$ changes. Let $\hat{f}_{-iK}$ be the estimator for model $M_K$ with observation $i$ removed. Using the derivation in Equation (\ref{eq:linearGCV}),
\begin{align}\notag
\frac{1}{n} \sum_{i=1}^n \left(y_i - \hat{f}_{-iK}(\x_i)\right)^2 &  = \frac{1}{n} \sum_{i=1}^n \left(y_i - \max_{k \in \{1,\dots,K\}} \alpha_k^{(-i)} + {\beta_k^{(-i)}}^T \x_i \right)^2,\\\notag
& = \frac{1}{n} \sum_{i=1}^n \left(\frac{y_i - \alpha_{k(i)} - \beta_{k(i)}^T \x_i}{1 - H^{k(i)}_{ii}\1_{\{i \in C_{k(i)}\}}}\right)^2\\\label{eq:gcv}
& \approx \frac{1}{n} \sum_{i=1}^n \left(\frac{y_i - \alpha_{k(i)} - \beta_{k(i)}^T \x_i}{1 - \frac{p+1}{|C_{k(i)}|}\1_{\{i \in C_{k(i)}\}}}\right)^2,
\end{align}where, in a slight abuse of notation, $H_{ii}^k$ is the diagonal entry of the hat matrix for subset $k$ corresponding to element $i$, and $$k(i) = \arg \max_{k \in \{1,\dots,K\}} \frac{\alpha_k + \beta_k^T \x_i }{1 - \frac{p+1}{|C_k|} \1_{\{i \in C_k\}}}.$$To select $K$, we find the $K$ that minimizes the right hand side of Equation (\ref{eq:gcv}). Although more computationally intensive than traditional generalized cross-validation, the computational complexity for CAP generalized cross-validation is similar to that of the CAP split selection step.

\subsection{Fast CAP}
The CAP algorithm offers two main computational bottlenecks. First, it searches over all cardinal directions, and only cardinal directions, to produce candidate models. Second, it keeps generating models until no subsets can be split without one having less than the minimum number of observations. In most cases, the optimal number of components is much lower than the terminal number of components.

To alleviate the first problem, we suggest using $P^{\prime}$ random projections as a basis for search. Using ideas similar to compressive sensing, each projection $\mathbf{g}_{j} \sim N_{p}(0,I)$ for $j = 1,\dots, P^{\prime}$. Then we search along the direction $\mathbf{g}_{j}^T \x$ rather than $x_j$. When we expect the true function to live in a lower dimensional space, as is the case with superfluous covariates, we can set $P^{\prime} < p$.

We solve the second problem by modifying the stopping rule. Instead of fulling growing the tree until each subset has less than $2n/(2\log(n))$ observations, we use generalized cross-validation. We grow the tree until the generalized cross-validation value has increased in two consecutive iterations or each subset has less than $2n/(2\log(n))$ observations. As the generalized cross-validation error is usually concave in $K$, this heuristic often offers a good fit at a fraction of the computational expense of the full CAP algorithm.

The Fast CAP algorithm has the potential to substantially reduce the $\log(n)\log( \log(n))$ factor by halting the model generation long before $K$ reaches $D\log(n)$. Since every feasible partition is searched for splitting, the computational complexity grows as $k$ gets larger.

The Fast CAP algorithm is summarized as follows.

\noindent \paragraph{Fast Convex Adaptive Partitioning (Fast CAP)}
\begin{enumerate}
	\item {\bf Initialize.} As in CAP.
	\item {\bf Split.}
	\begin{enumerate}
		\item[a.] {\it Generate candidate splits.} Generate candidate model $\hat{M}_{kj\ell}$ by 1) fixing a subset $k$, 2) generating a random direction $j$ with $\mathbf{g}_j \sim N_p(0,I)$, and 3) dyadically dividing the data as follows:
		\begin{itemize}
			\item set $x_{min}^{kj} = \min \{ \mathbf{g}_j^T\x_i : i \in C_k\}$, $x_{max}^{kj} = \max \{ \mathbf{g}_j^T\x_i : i \in C_k\}$ and $b_{kj\ell} = a_{\ell} x_{min}^{kj} + (1- a_{\ell})x_{max}^{kj}$
			\item set
		\begin{align}\notag
\hat{C}_k & = \{i : i \, \in C_k, \, \mathbf{g}_j^Tx_{i} \leq b_{kj\ell}\}, & \hat{C}_{K+1} & = \{i : i \, \in C_k, \, \mathbf{g}_j^T \x_{i} > b_{kj\ell}\},\\\notag
\hat{A}_k & = \{\x : \x \, \in A_k, \, \mathbf{g}_j^T \x \leq b_{kj\ell}\}, & \hat{A}_{K+1} & = \{\x : \x \, \in A_k, \, \mathbf{g}_j^T \x > b_{kj\ell}\}.
\end{align}
		\end{itemize}
	 	Then new hyperplanes are fit to each of the new subsets. This is done for $L$ knots, $P^{\prime}$ dimensions and $K$ subsets.
		\item[b.] {\it Select split.} As in CAP.
	\end{enumerate}
	\item {\bf Refit.} As in CAP.
	\item {\bf Stopping conditions.} Let $GCV(M_K)$ be the generalized cross-validation error for model $M_K$. Stop if $GCV(M_{K}) > GCV(M_{K-1})$ and $GCV(M_{K-1}) > GCV(M_{K-2})$. Then select final model as in CAP.
\end{enumerate}

%=====================================================
%  Applications
%=====================================================

\section{Applications}\label{sec:numbers}

In this section, we empirically analyze the performance of CAP. There are no benchmark problems for multivariate convex regression, so we analyze the predictive performance, runtime, sensitivity to tunable parameters and rates of convergence on a set of synthetic problems. We then apply CAP to value function approximation for pricing American basket options.

%The questions to be answered:
%\begin{itemize}
%	\item How well does CAP do compared to existing convex regression methods (when applicable)?
%	\item How well does CAP do compared to nonparametric regression methods?
%	\item How long does it take to run?  How does it scale with $d$ and $n$?
%	\item What about the partitioning level?
%\end{itemize}

%-----------------------------------------------------------------------------
%  Synthetic Problems
%-----------------------------------------------------------------------------

\subsection{Synthetic Regression Problems}

We apply CAP to two synthetic regression problems to demonstrate predictive performance and analyze sensitivity to tunable parameters. The first problem has a non-additive structure, high levels of covariate interaction and moderate noise, while the second has a simple univariate structure embedded in a higher dimensional space and low noise. Low noise or noise free problems often occur when a highly complicated convex function needs to be approximated by a simpler one~\citep{MaBo09}.

\paragraph{Problem 1}
Here $\x \in \R^5$.  Set $$y = \left(x_1+.5x_2+x_3\right)^2 - x_4+.25x_5^2 + \epsilon,$$
where $\epsilon \sim N(0,1)$. The covariates are drawn from a 5 dimensional standard Gaussian distribution, $N_5(0,I)$.

\paragraph{Problem 2}
Here $\x \in \R^{10}$.  Set $$y = \exp\left( \x^T \mathbf{p}\right)+ \epsilon,$$
where $\mathbf{p}$ was randomly drawn from a Dirichlet(1,$\dots$,1) distribution, $$\mathbf{p} = (0.0680, 0.0160, 0.1707,0.1513,0.1790,0.2097, 0.0548,0.0337,0.0377,0.0791)^T.$$  We set $\epsilon \sim N(0,0.1^2)$. The covariates are drawn from a 10 dimensional standard Gaussian distribution, $N_{10}(0,I)$.

% Runtime and predictive error

\subsubsection{Predictive Performance and Runtimes}

We compared the performance of CAP and Fast CAP to other regression methods on Problems 1 and 2. The only other convex regression method included was least squares regression (LSE); it was implemented with the {\tt cvx} convex optimization solver. The general methods included Gaussian processes~\citep{RaWi06}, a widely implemented Bayesian nonparametric method, and two adaptive methods:  tree regression with constant values in the leaves and Multivariate Adaptive Regression Splines (MARS)~\citep{Fr91}. Tree regression was run through the Matlab function {\tt classregtree}. MARS was run through the Matlab package {\tt ARESLab}. Gaussian processes were run with the Matlab package {\tt gpml}.

%Tree regression was run through the Matlab function {\tt classregtree}. Gaussian processes were run with the Matlab package {\tt gpml}.

%All linear tree models generated potential subset splits in the same manner as CAP. 

Parameters for CAP and Fast CAP were set as follows. The log scale parameter set as $D = 3$ and the number of knots was set as $L = 10$ for both. In Fast CAP, the number of random search directions was set to be $p$. 

All methods were given a maximum runtime of 90 minutes, after which the results were discarded. Methods were run on 10 random training sets and tested on the same testing set. Average runtimes and predictive performance are given in Table \ref{tab:synthetic}.

\begin{table}
\footnotesize
\begin{center}
\begin{tabular}{l | r@{.}l | r@{.}l | r@{.}l | r@{.}l | r@{.}l | r@{.}l  | r@{.}l }
\multicolumn{15}{c}{Mean Squared Error}\\
\multicolumn{15}{c}{Problem 1}\\
Method & \multicolumn{2}{c |}{$n = 100$} & \multicolumn{2}{c |}{$n = 200$} & \multicolumn{2}{c |}{$n = 500$} & \multicolumn{2}{c |}{$n = 1,000$} & \multicolumn{2}{c |}{$n = 2,000$} & \multicolumn{2}{c |}{$n = 5,000$} & \multicolumn{2}{c }{$n = 10,000$}\\\hline
 CAP & {\bf1} & {\bf5884} & {\bf0} & {\bf6827} & {\bf0} & {\bf2740} & 0 & 1644 &{\bf0} & {\bf0927} & {\bf0} & {\bf0629} & {\bf0} & {\bf0450}\\
  Fast CAP & 1 & 8661 & 0 & 7471 & 0 & 3197 & {\bf0} & {\bf1526} & 0 & 1356 & 0 & 0724 & 0 & 0566\\
 %L/NR Linear Tree & 5 & 0036 & 2 & 2614 & 1 & 0846 & 0 & 7140 & 0 & 4677 & 0 & 3382 & 0 & 2373 \\
 %L/R Linear Tree & 2 & 5696 & 0 & 9283 & 0 & 7308 & 0 & 2730 & 0 & 1261 & 0 & 0853 & 0 & 0712 \\
 %G/NR Linear Tree & 3 & 7693 & 1 & 4733 & 0 & 8535 & 0 & 5465 & 0 & 3302 & 0 & 2351 & 0 & 1917 \\
 LSE & 15 & 8340 & 9 & 5970 & 18 & 0701 & 9,862 & 4602 & \multicolumn{2}{c|}{ -- } & \multicolumn{2}{c|}{ -- } & \multicolumn{2}{c}{ -- } \\
 Tree & 12 & 2794 & 9 & 8356 & 6 & 7606 & 5 & 3478 & 4 & 1230 & 2 & 9173 & 2 & 3152 \\
 GP & 8 & 5056 & 13 & 5495 & 6 & 8472 & 3 & 7610 & 2 & 2928 & 1 & 2058 & \multicolumn{2}{c}{ -- }  \\
 MARS & 8 & 3517 & 8 & 0031 & 6 & 8813 & 6 & 2618 & 5 & 9809 & 5 & 8558 & 5 & 8234 \\
\hline
 \multicolumn{15}{c}{Problem 2}\\
Method & \multicolumn{2}{c |}{$n = 100$} & \multicolumn{2}{c |}{$n = 200$} & \multicolumn{2}{c |}{$n = 500$} & \multicolumn{2}{c |}{$n = 1,000$} & \multicolumn{2}{c |}{$n = 2,000$} & \multicolumn{2}{c |}{$n = 5,000$} & \multicolumn{2}{c }{$n = 10,000$}\\\hline
CAP & 0 & 0159 & 0 & 0138 & 0 & 0110 & {\bf0} & {\bf0018} & 0 & 0012 & {\bf0} & {\bf0007} & {\bf0} & {\bf0003} \\
Fast CAP & 0 & 0159 & 0 & 0138 & 0 & 0090 & {\bf0} & {\bf0018} & {\bf0} & {\bf0011} & {\bf0} & {\bf0007} & {\bf0} & {\bf0003} \\
%L/NR Linear Tree & 0 & 0159 & 0 & 0138 & 0 & 0094 & 0 & 0064 & 0 & 0049 & 0 & 0037 & 0 & 0022 \\
%L/R Linear Tree & 0 & 0159 & 0 & 0138 & 0 & 0111 & 0 & 0018 & 0 & 0011 & 0 & 0008 & 0 & 0005 \\
%G/NR Linear Tree & 0 & 0159 & 0 & 0138 & 0 & 0098 & 0 & 0052 & 0 & 0041 & 0 & 0028 & 0 & 0019 \\
 LSE & 0 & 6286 & 0 & 2935 & 31 & 2426 & \multicolumn{2}{c|}{ -- } & \multicolumn{2}{c|}{ -- } & \multicolumn{2}{c|}{ -- } & \multicolumn{2}{c}{ -- } \\
 Tree & 0 & 1372 & 0 & 1129 & 0 & 0928 & 0 & 0797 & 0 & 0670 & 0 & 0552 & 0 & 0495 \\
 GP & {\bf 0} & {\bf 0109} & {\bf0} & {\bf0063} & {\bf0} & {\bf0039} & 0 & 0027 & 0 & 0047 & 0 & 0076 & \multicolumn{2}{c}{ -- } \\
 MARS & 0 & 0205 & 0 & 0140 & 0 & 0120 & 0 & 0110 & 0 & 0105 & 0 & 0102 & 0 & 0100 \\
\hline
\multicolumn{15}{c}{}\\
\multicolumn{15}{c}{Run Time}\\
\multicolumn{15}{c}{Problem 1}\\
Method & \multicolumn{2}{c |}{$n = 100$} & \multicolumn{2}{c |}{$n = 200$} & \multicolumn{2}{c |}{$n = 500$} & \multicolumn{2}{c |}{$n = 1,000$} & \multicolumn{2}{c |}{$n = 2,000$} & \multicolumn{2}{c |}{$n = 5,000$} & \multicolumn{2}{c }{$n = 10,000$}\\\hline
 CAP & 0 & 15 sec & 0 & 24 sec & 0 & 78 sec & 1 & 34 sec & 2 & 18 sec & 4 & 33 sec & 9 & 31 sec\\
 Fast CAP & 0 & 04 sec & 0 & 07 sec & 0 & 15 sec & 0 & 30 sec & 0 & 57 sec & 1 & 14 sec & 2 & 06 sec\\
 %L/NR Linear Tree & 0 & 06 sec & 0 & 21 sec & 0 & 89 sec & 1 & 75 sec & 2 & 51 sec & 4 & 55 sec & 11 & 85 sec \\
 %L/R Linear Tree & 0 & 06 sec & 0 & 22 sec & 0 & 73 sec & 1 & 30 sec & 1 & 75 sec & 3 & 46 sec & 7 & 18 sec \\
 %G/NR Linear Tree & 0 & 08 sec & 0 & 20 sec & 0 & 91 sec & 1 & 71 sec & 2 & 66 sec & 5 & 07 sec & 12 & 17 sec \\
 LSE & 1 & 56 sec & 10 & 17 sec & 226 & 20 sec & 43 & 37 min & \multicolumn{2}{c|}{ -- } & \multicolumn{2}{c|}{ -- } & \multicolumn{2}{c}{ -- } \\
 Tree & 0 & 06 sec & 0 & 02 sec & 0 & 04 sec & 0 & 09 sec & 0 & 19 sec & 0 & 49 sec & 1 & 15 sec \\
 GP & 0 & 22 sec & 0 & 35 sec & 1 & 35 sec & 5 & 07 sec & 22 & 03 sec & 248 & 72 sec & \multicolumn{2}{c}{ -- } \\
 MARS & 0 & 22 sec & 0 & 34 sec & 0 & 76 sec & 1 & 81 sec & 3 & 95 sec & 16 & 65 sec & 56 & 19 sec \\
 \hline
\multicolumn{15}{c}{Problem 2}\\
Method & \multicolumn{2}{c |}{$n = 100$} & \multicolumn{2}{c |}{$n = 200$} & \multicolumn{2}{c |}{$n = 500$} & \multicolumn{2}{c |}{$n = 1,000$} & \multicolumn{2}{c |}{$n = 2,000$} & \multicolumn{2}{c |}{$n = 5,000$} & \multicolumn{2}{c }{$n = 10,000$}\\\hline
CAP & 0 & 05 sec & 0 & 25 sec & 2 & 15 sec & 6 & 35 sec & 10 & 06 sec & 21 & 06 sec & 46 & 50 sec \\
Fast CAP & 0 & 02 sec & 0 & 03 sec & 0 & 08 sec & 0 & 13 sec & 0 & 25 sec & 0 & 89 sec & 2 & 03 sec \\
%L/NR Linear Tree & 0 & 03 sec & 0 & 11 sec & 1 & 01 sec & 2 & 52 sec & 4 & 17 sec & 10 & 31 sec & 24 & 30 sec \\
%L/R Linear Tree & 0 & 03 sec & 0 & 12 sec & 1 & 17 sec & 2 & 41 sec & 3 & 60 sec & 9 & 33 sec & 19 & 50 sec \\
%G/NR Linear Tree & 0 & 03 sec & 0 & 12 sec & 1 & 34 sec & 3 & 66 sec & 5 & 12 sec & 11 & 63 sec & 26 & 46 sec \\
 LSE & 1 & 86 sec & 15 & 13 sec & 339 & 16 sec & \multicolumn{2}{c|}{ -- } & \multicolumn{2}{c|}{ -- } & \multicolumn{2}{c|}{ -- } & \multicolumn{2}{c}{ -- } \\
 Tree & 0 & 02 sec & 0 & 03 sec & 0 & 07 sec & 0 & 14 sec & 0 & 27 sec & 0 & 71 sec & 1 & 53 sec \\
 GP & 0 & 15 sec & 0 & 34 sec & 1 & 46 sec & 4 & 93 sec & 23 & 13 sec & 264 & 77 sec & \multicolumn{2}{c}{ -- } \\
 MARS & 0 & 72 sec & 0 & 48 sec & 1 & 38 sec & 3 & 43 sec & 8 & 01 sec & 33 & 29 sec & 98 & 75 sec \\
\hline
\end{tabular}
\end{center}
\normalsize
\caption{CAP, Fast CAP, Least Squares Estimator (LSE), Gaussian Processes (GP), MARS and tree regression were run on Problems 1 and 2. Errors are in distance to the true mean function. The lowest error values are bolded.}\label{tab:synthetic}
\end{table}

Unsurprisingly, the non-convex regression methods did poorly compared to CAP and Fast CAP, particularly in the higher noise setting. Gaussian processes offered the best performance of that group, but their computational complexity scales like $\bigO(n^3)$; this computational times of more than 90 minutes for $n = 10,000$. More surprisingly, however, the LSE did extremely poorly. This can be attributed to overfitting, particularly in the boundary regions; this phenomenon can be seen in Figure \ref{fig:CAPvLSE} as well. While the natural response to overfitting is to apply a regularization penalty to the hyperplane parameters, implementation in this setting is not straightforward. We have tried implementing $\ell_2$ penalties on the hyperplane coefficients, but tuning the parameters quickly became computationally infeasible due to runtime issues with the LSE.

%We implemented variations of tree models to study the benefits of CAP's two stage partition refinement, which 1) uses a global error metric to select a split, and 2) uses the hyperplanes to refit the partition. Results show that the refitting step substantially reduces error while a global error metric adds a modest reduction. The combination of the two seems to be necessary to produce a highly reliable, low error algorithm. Predictive results were comparable between CAP and Fast CAP.

%When the results for the methods that include a convexity restriction (CAP, Fast CAP, LSE) are compared to those without (GP, Constant Tree), we see that inclusion of a convexity constraint is associated with dramatic error reduction, particularly in noisy settings. This is likely because convexity acts as a strong regularizer.

% Why did MARS and CART fail?
% Why did GP fail?
% - O(n^3) computation makes prediction slow on larger problems
% - cannot use convexity as a regularizer
% Why did LSE fail?
% - surprisingly poor predictive error, especially in MSE (compared to MAE)
% - reeeaaallly sloooooooow
% Why are runtimes not always increasing?

% Runtimes
Although CAP and Fast CAP had similar predictive performance, their runtimes often differed by an order of magnitude with the largest differences on the biggest problem sizes. Based on this performance, we would suggest using Fast CAP on larger problems rather than the full CAP algorithm.

% The runtimes of Fast CAP were comparable to those of {\tt classregtree}, an optimized, {\tt Matlab} specific, piecewise constant tree regression procedure. 

%We also note that we have found substantial increases in computational time in {\tt Matlab} when a procedure is called for the first time. This likely accounts for the relatively high computational times on the smallest problems, which were called first.

\subsubsection{CAP and Treed Linear Models}\label{sec:linearModels}

% Things that are different:
%  - global split selection: reduces some error, adds a lot of stability
%  - refitting reduces a lot of error

Treed linear models are a popular method for regression and classification. They can be easily modified to produce a convex regression estimator by taking the maximum over all of the linear models. CAP differs from existing treed linear models in how the partition is refined. First, subset splits are selected based on global reduction of error. Second, the partition is refit after a split is made. To investigate the contributions of each step, we compare to treed linear models generated by: 1) local error reduction as an objective for split selection and no refitting, 2) global error reduction as an objective function for split selection and no refitting, and 3) local error reduction as an objective for split selection along with refitting. All estimators based on treed linear models are generated by taking the maximum over the set of linear models in the leaves.

We wanted to determine which properties led to a low variance estimator with low predictive error. By low variance, we mean that changes in the training set do not lead to large changes in predictive error. To do this, we compared the performance of these methods on Problems 1 and 2 over 10 different training sets and a single testing set. All treed linear model parameters were the same as those for CAP. We viewed a model with local subset split selection and no refitting as a baseline. We compared both the average squared predictive error and the variance of that error between training sets. Percentages of average error and variance reduction are displayed in Table \ref{tab:treedLinearModel}. Average predictive error is displayed in Figure \ref{fig:treedLinearModel}.

\begin{table}
\footnotesize
\begin{center}
\begin{tabular}{l | r@{.}l | r@{.}l | r@{.}l | r@{.}l | r@{.}l | r@{.}l  | r@{.}l }
\multicolumn{15}{c}{Percentage Reduction in Mean Squared Error}\\
\multicolumn{15}{c}{Problem 1}\\
Method & \multicolumn{2}{c |}{$n = 100$} & \multicolumn{2}{c |}{$n = 200$} & \multicolumn{2}{c |}{$n = 500$} & \multicolumn{2}{c |}{$n = 1,000$} & \multicolumn{2}{c |}{$n = 2,000$} & \multicolumn{2}{c |}{$n = 5,000$} & \multicolumn{2}{c }{$n = 10,000$}\\\hline
Refitting & 48 & 65\% & \ 58 & 95\% & 32 & 62\% & \ \ \ 61 & 76\% & \ \ \ 73 & 04\% & \ \ \ 74 & 77\% & \ \ \ \ 70 & 01 \% \\
Global Selection & 24 & 67\% & 34 & 85\% & 21 & 32\% & 23 & 46\% & 29 & 40\% & 30 & 48\% & 19 & 23\% \\
CAP & 68 & 25\% & 69 & 81\% & 74 & 74\% & 76 & 97\% & 80 & 18\% & 81 & 40\% & 81 & 04\% \\
\hline
\multicolumn{15}{c}{Problem 2}\\
Method & \multicolumn{2}{c |}{$n = 100$} & \multicolumn{2}{c |}{$n = 200$} & \multicolumn{2}{c |}{$n = 500$} & \multicolumn{2}{c |}{$n = 1,000$} & \multicolumn{2}{c |}{$n = 2,000$} & \multicolumn{2}{c |}{$n = 5,000$} & \multicolumn{2}{c }{$n = 10,000$}\\\hline
Refitting & 0 & 0\% & 0 & 0\% & -17 & 73\% & 71 & 48\% & 78 & 36\% & 79 & 67\% & 77 & 05\% \\
Global Selection & 0 & 0\% & 0 & 0\% & -4 & 36\% & 17 & 69\% & 15 & 22\% & 25 & 04\% & 9 & 74\% \\
CAP & 0 & 0\% & 0 & 0\% & -17 & 10\% & 71 & 70\% & 75 & 60\% & 81 & 66\% & 86 & 21\% \\
\hline
\multicolumn{15}{c}{}\\
\multicolumn{15}{c}{Percentage Reduction in Variance of Mean Squared Error}\\
\multicolumn{15}{c}{Problem 1}\\
Method & \multicolumn{2}{c |}{$n = 100$} & \multicolumn{2}{c |}{$n = 200$} & \multicolumn{2}{c |}{$n = 500$} & \multicolumn{2}{c |}{$n = 1,000$} & \multicolumn{2}{c |}{$n = 2,000$} & \multicolumn{2}{c |}{$n = 5,000$} & \multicolumn{2}{c }{$n = 10,000$}\\\hline
Refitting & 19 & 16\% & 65 & 00\% & -243 & 33\% & -4 & 03\% & -163 & 40\% & 64 & 86\% & -18 & 88\% \\
Global Selection & 38 & 41\% & 68 & 78\% & -17 & 84\% & 61 & 34\% & 24 & 51\% & 91 & 44\% & 75 & 97\%\\
CAP & 96 & 89\% & 92 & 72\% & 68 & 74\% & 97 & 05\% & 74 & 85\% & 95 & 29\% & 63 & 17\%\\
\hline
\multicolumn{15}{c}{Problem 2}\\
Method & \multicolumn{2}{c |}{$n = 100$} & \multicolumn{2}{c |}{$n = 200$} & \multicolumn{2}{c |}{$n = 500$} & \multicolumn{2}{c |}{$n = 1,000$} & \multicolumn{2}{c |}{$n = 2,000$} & \multicolumn{2}{c |}{$n = 5,000$} & \multicolumn{2}{c }{$n = 10,000$}\\\hline
Refitting & 0 & 0\% & 0 & 0\% & -61 & 34\% & 44 & 75\% & 94 & 16\% & 73 & 93\% & 75 & 42\% \\
Global Selection & 0 & 0\% & 0 & 0\% & -19 & 84\% & -223 & 58\% & -209 & 92\% & -8 & 29\% & -7 & 17\% \\
CAP & 0 & 0\% & 0 & 0\% & -76 & 78\% & 52 & 44\% & 89 & 30\% & 30 & 18\% & 15 & 16\% \\
\hline
\end{tabular}
\end{center}
\caption{Percentage reduction in mean squared error and variance of mean squared error compared to treed linear model with local split selection and no refitting on Problems 1 and 2 over 10 training sets. This model was compared to 1) a treed linear model with partition refitting but local split selection, 2) a treed linear model with global split selection but no partition refitting, and 3) CAP.}\label{tab:treedLinearModel}
\end{table}

\begin{figure}
\begin{center}
\includegraphics[width=6.5in]{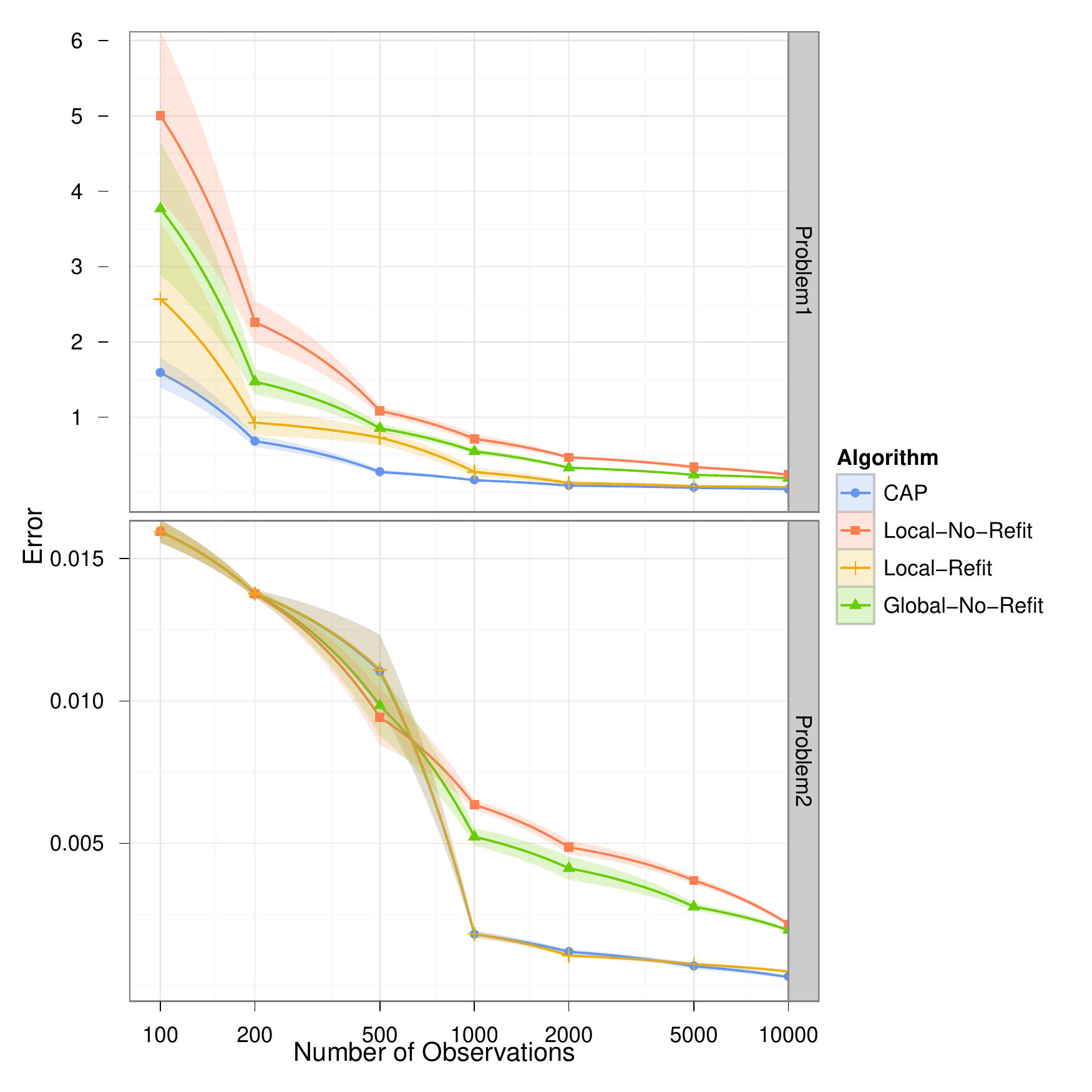}
\end{center}
\caption{Number of observations (log scale) vs. mean squared error (linear scale) for CAP and treed linear models with 1) local split selection, with no refitting, 2) local split selection, with refitting, and 3) global split selection, with no refitting. Mean error plus/minus one standard deviation are shown for data taken from 10 training sets.}
\label{fig:treedLinearModel}
\end{figure}

Table \ref{tab:treedLinearModel} shows that global split selection and refitting are both beneficial, but in different ways. Refitting dramatically reduces predictive error, but can variance to the estimator in noisy settings. Global split selection modestly reduces predictive error but can reduce variance in noisy settings, like Problem 1. The combination of the two produces CAP, which has both low variance and high predictive accuracy.

\subsubsection{Sensitivity to Tunable Parameters}\label{sec:sensitivity}

% Number of knots:
% - L = [1,2,5,10,15,20,50], n = [500,5000]
% - Run with CAP, Fast CAP
% Log Parameter
% - D = [1/5, 1/3, 1, 3, 5, 10, 20], n = [500,5000]
% - Run with CAP, Fast CAP
% 	+ D changes runtime! (Particularly for CAP)
%	+ must balance runtime with better accuracy

In this subsection, we empirically examine the effects of the two tunable parameters, the log factor, $D$, and the number of knots, $L$. The log factor controls the minimal number of elements in each subset by setting $|C_k| \geq n/(D\log(n))$, and hence it controls the number of subsets, $K$, at least for large enough $n$. Increasing $D$ allows the potential accuracy of the estimator to increase, but at the cost of greater computational time due to the increase in possible values for $K$ and the larger number of possibly admissible sets generated in the splitting step of CAP.

\begin{figure}
\begin{center}
\includegraphics[width=6.5in, viewport = 20 260 570 535]{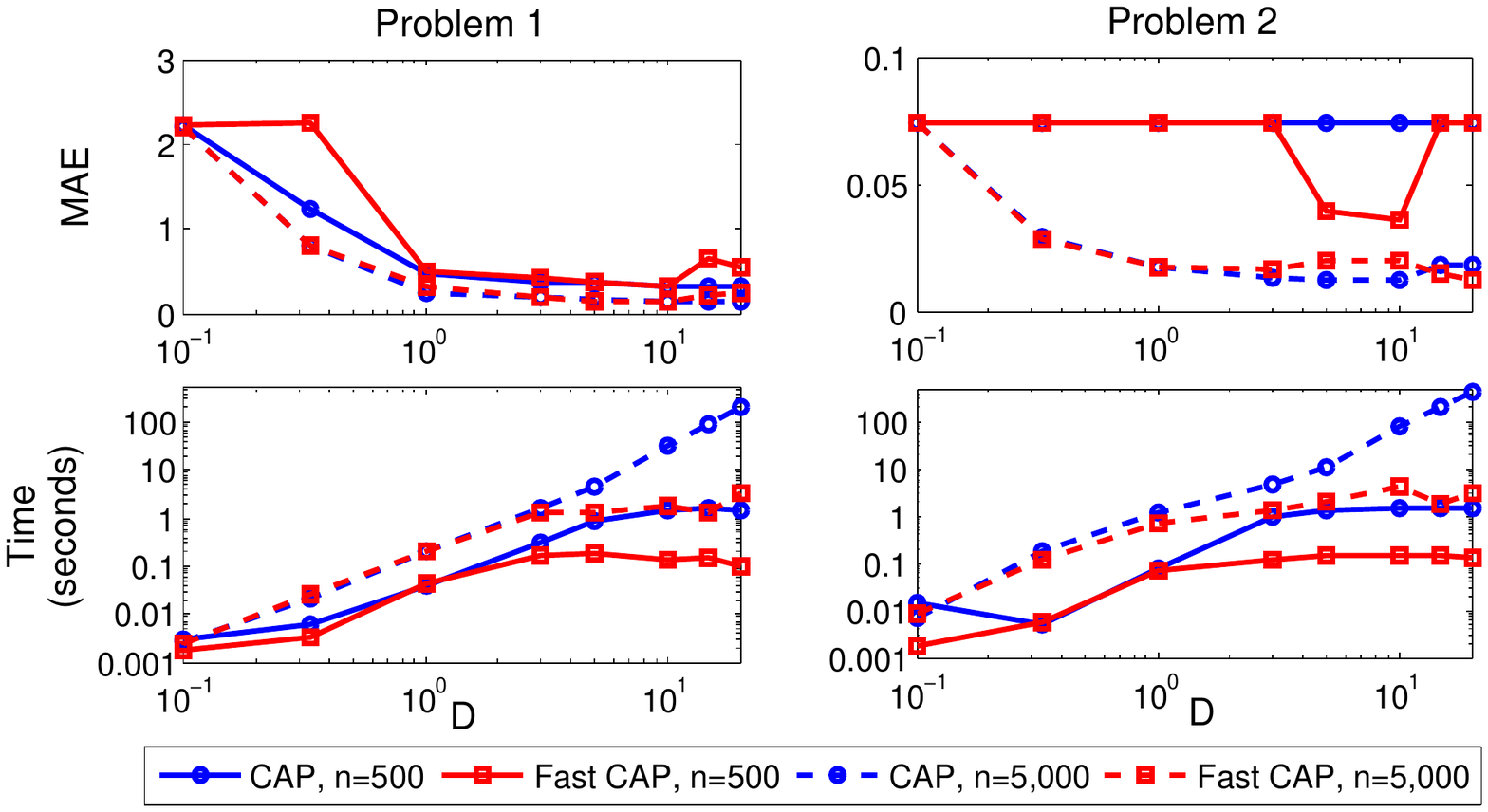}
\end{center}
\caption{(Top) Log factor $D$ (log scale) vs. mean absolute error (linear scale) for CAP and Fast CAP.  (Bottom) Log factor $D$ (log scale) vs. runtime in seconds (log scale). Both methods were run on Problem 1 (left) and Problem 2 (right) with $n = 500$ and $n = 5,000$.}
\label{fig:Dcompare}
\end{figure}

We compared values for $D$ ranging from $0.1$ to $20$ on Problems 1 and 2 with sample sizes of $n = 500$ and $n = 5,000$. Results are displayed in Figure \ref{fig:Dcompare}. Note that error may not be strictly decreasing with $D$ because different subsets are proposed under each value. Additionally, Fast CAP is a randomized algorithm so variance in error rate and runtime is to be expected. 

Empirically, once $D \geq 1$, there was little substantive error reduction in the models, but the runtime increased as $\bigO(D^2)$ for the full CAP algorithm. Since $D$ controls the maximum partition size, $K_n = D\log(n)$, and a linear regression is fit $K\log(K)$ times, the expected increase in the runtime should only be $\bigO(D\log(D))$. We believe that the extra empirical growth comes from an increased number of feasible candidate splits. In the Fast CAP algorithm, which terminates after generalized cross-validation gains cease to be made, we see runtimes leveling off with higher values of $D$. Based on these results, we believe that setting $D=3$ offers a good balance between fit and computational expense.

\begin{figure}
\begin{center}
\includegraphics[width=6.5in, viewport = 20 260 570 535]{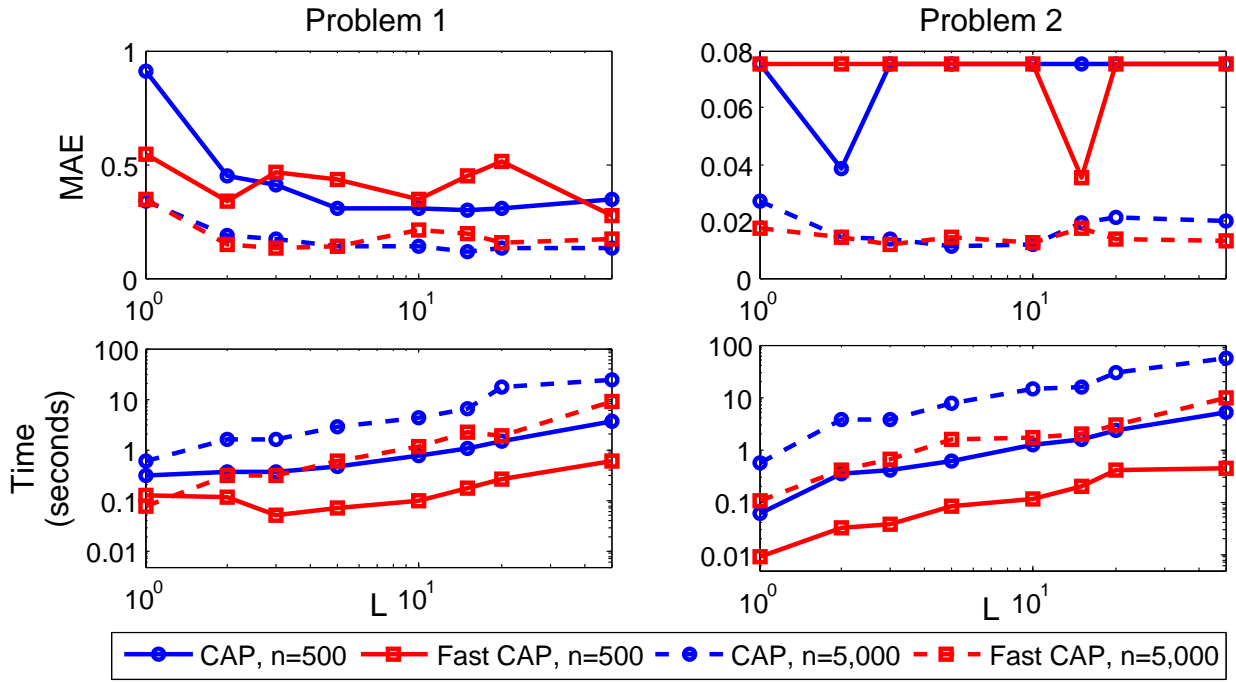}
\end{center}
\caption{(Top) Number of knots $L$ (log scale) vs. mean absolute error (linear scale) for CAP and Fast CAP.  (Bottom) Number of knots $L$ (log scale) vs. runtime in seconds (log scale). Both methods were run on Problem 1 (left) and Problem 2 (right) with $n = 500$ and $n = 5,000$.}
\label{fig:Lcompare}
\end{figure}

The number of knots, $L$, determines how many possible subsets will be examined during the splitting step. Like $D$, an increase in $L$ offers a better fit at the expense of increased computation. We compared values for $D$ ranging from $1$ to $50$ on Problems 1 and 2 with sample sizes of $n = 500$ and $n = 5,000$. Results are displayed in Figure \ref{fig:Lcompare}.

The changes in fit and runtime are less dramatic with $L$ than they are with $D$. After $L = 3$, the predictive error rates almost completely stabilized. Runtime increased as $\bigO(L)$ as expected. Due to the minimal increase in computation, we feel that $L = 10$ is a good choice for most settings.

\subsubsection{Empirical Rates of Convergence}
Although theoretical rates of convergence are not yet available for CAP, we are able to empirically examine them. Rates of convergence for multivariate convex regression have only been studied in two articles of which we are aware. First, \citet{AgFoMo11} studied rates of convergence for an estimator that is created by first smoothing the data, then evaluating the smoothed data over an $\epsilon$-net, and finally convexifying the net of smoothed data by taking the convex hull. They showed that the convexify step preserved the rates of the smoothing step. For most smoothing algorithms, these are minimax nonparametric rates, $n^{-\frac{1}{p+2}}$ with respect to the empirical $\ell_2$ norm. In the second article, \citet{HaDu11c} showed adaptive rates for a Bayesian model that places a prior over the set of all piecewise linear functions. Specifically, they showed that if the true mean function $f_0$ actually maps a $d$-dimensional linear subspace of $\mathcal{X}$ to $\R$, that is
\begin{align}\notag
f_0(\x) & = g_0(\x \A), & \A & \in \R^{p\times d},
\end{align}then their model achieves rates of $\log^{-1}(n) n^{-\frac{1}{d+2}}$ with respect to the empirical $\ell_2$ norm. Empirically, we see these types of adaptive rates with CAP.

% Summary of rate literature:
% - MBCR: general nonparametric and adaptive
% - AgFoMo11: convexify step does not change rates

\begin{figure}
\begin{center}
\includegraphics[width=6in, viewport = 0 260 620 535]{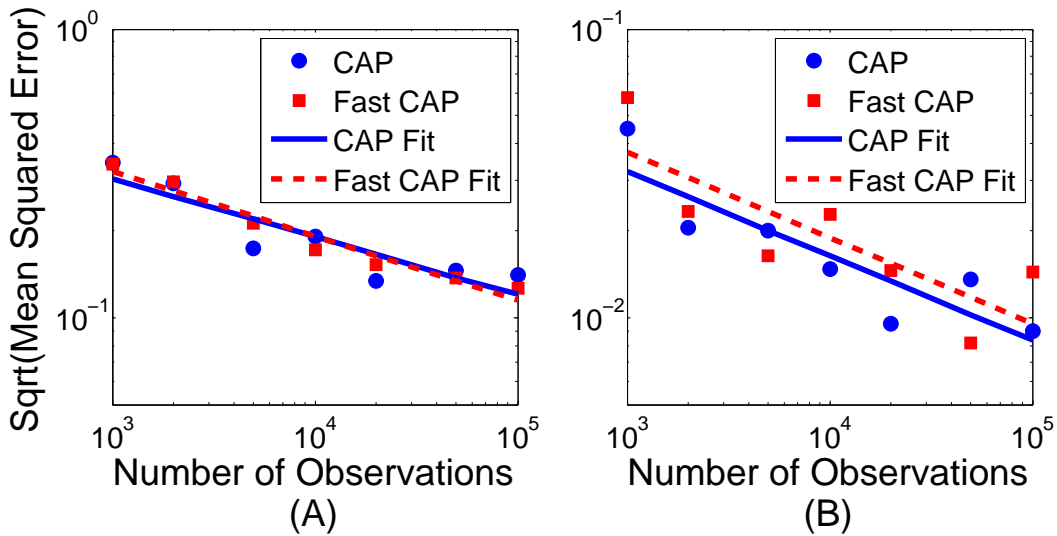}
\end{center}
\caption{Number of observations $n$ (log scale) vs. square root of mean squared error (log scale) for (A) Problem 1 and (B) Problem 2. Linear models are fit to find the empirical rate of convergence.}
\label{fig:rates}
\end{figure}

\begin{table}
\begin{center}
\begin{tabular}{l | r@{.}l | r@{.}l }
 Method & \multicolumn{2}{ c  }{Problem 1} & \multicolumn{2}{|c}{Problem 2}\\\hline
 %Method & \multicolumn{2}{| c }{Slope} & \multicolumn{2}{| c  }{Intercept} & \multicolumn{2}{| c  }{Slope} & \multicolumn{2}{| c }{Intercept}\\\hline
 Expected: Rates in $p$ & $-0$ & 1429 & $ -0$ & 0833\\
 Expected: Rates in $d$ & $-0$ & 2000 & $ -0$ & 3333 \\
 \hline
 Empirical: CAP & $ -0$ & 2003 & $-0$ & 2919 \\
 Empirical: Fast CAP & $-0$ & 2234 & $-0$ & 2969 \\\hline
\end{tabular}
\end{center}
\normalsize
\caption{Slopes for linear models fit to $\log(n)$ vs. $\log(\sqrt{MSE})$ in Figure \ref{fig:rates}. Expected slopes are given when: 1) rates are with respect to full dimensionality, $p$, and 2) rates are with respect to dimensionality of linear subspace, $d$. Empirical slopes are fit to mean squared error generated by CAP and Fast CAP. Note that all empirical slopes are closest to those for linear subspace rates rather than those for full dimensionality rates.}
\label{tab:rates}
\end{table}

In Figure \ref{fig:rates}, we plotted the number of observations against the square root of the mean squared error in a log-log plot for Problems 1 and 2. We then fitted a linear model for both CAP and Fast CAP. For Problem 1, $p=5$ but $d =3$, due to the sum in the quadratic term. Likewise, for Problem 2, $p = 10$ but $d = 1$ because it is an exponential of a linear combination. Under standard nonparametric rates, we would expect the slope of the linear model to be $-\frac{1}{7}$ for Problem 1 and $-\frac{1}{12}$ for Problem 2. However, we see slopes closer to $-\frac{1}{5}$ and $ -\frac{1}{3}$ for Problems 1 and 2, respectively; values are given in Table \ref{tab:rates}. These results strongly imply that CAP achieves adaptive convergence rates of the type shown by \citet{HaDu11c} for Problems 1 and 2.

%-----------------------------------------------------------------------------
%  Real-life Problems
%-----------------------------------------------------------------------------

\subsection{Pricing Stock Options}
In sequential decision problems, a decision maker takes an action based on a currently observed state of the world based on the current rewards of that action and possible future rewards. Approximate dynamic programming is a modeling method for such problems based on approximating a value-to-go function. Value-to-go functions, or simply ``value functions,'' give the value for each state of the world if all optimal decisions are made subsequently. 

Often value functions are known to be convex or concave in the state variable; this is common in options pricing, portfolio optimization and logistics problems.  In some situations, such as when a linear program is solved each time period to determine an action, a convex value function {\it is required} for computational tractability. Convex regression holds great promise for value function approximation in these problems.

To give a simple example for value function approximation, we consider pricing American basket options on the average of $N$ underlying assets.  Options give the holder the right---but not the obligation---to buy the underlying asset, in this case the average of $N$ individual assets, for a predetermined strike price $K$.  In an American option, this can be done at any time between the issue date and the maturity date, $T$.  However, American options are notoriously difficult to price, particularly when the underlying asset base is large.

A popular method for pricing American options uses approximate dynamic programming where continuation values are approximated via regression~\citep{Ca96,TsVa99,TsVa01,LoSc01}.  We summarize these methods as follows; see \citet{Gl04} for a more thorough treatment.  The underlying assets are assumed to have the sample path $\{X_1,\dots,X_T\}$, where $X_t = \{S_1(t),\dots,S_N(t)\}$ is the set of securities at time $t$.  At each time $t$, a continuation value function, $\bar{V}_{t}(X_{t})$, is estimated by regressing a value function for the next time period, $\bar{V}_{t+1}(X_{t+1})$, on the current state, $X_t$. The continuation value is the value of holding the option rather than exercising given the current state of the assets. The value function is defined to be the max of the current exercise value and the continuation value. Options are exercised when the current exercise value is greater than or equal to the continuation value.

The procedure to estimate the continuation values is as follows (as summarized in \citet{Gl04}):
\begin{enumerate}
	\item[0.] Define basket payoff function, $$h(X_t) = \max \left\{ \frac{1}{N} \sum_{k=1}^N S_k(T)-K,0\right\}.$$
	\item[1.] Sample $M$ independent paths, $\{X_{1j},\dots,X_{Tj}\},$ $j = 1,\dots,M$.
	\item[2.] At time $T$, set $\bar{V}_T(X_{Tj}) = h(X_{Tj})$.
	\item[3.] Apply backwards induction: for $t = T-1, \dots,1$,
	\begin{itemize}
		\item given $\{\bar{V}_{t+1}(X_{t+1 j})\}_{j=1}^M$, regress on $\{X_{tj}\}_{j=1}^M$ to get continuation value estimates $\{\bar{C}_t(X_{tj})\}_{j=1}^M$.
		\item set value function,
		\begin{equation}\notag
		\bar{V}_t(X_{tj}) = \max\left\{h(X_{tj}),\bar{C}_t(X_{tj})\right\}.\end{equation}	
		\end{itemize}
\end{enumerate}We use the value function defined by \citet{TsVa99}.

The regression values are used to create a policy that is implemented on a test set: exercise when the current exercise value is greater than or equal to the estimated continuation value. A good regression model is crucial to creating a good policy.

%the regression for the continuation values is used with new testing sample paths. Note that the estimated continuation values determine an exercise policy for the option.  As this is by definition a suboptimal exercise policy, it provides a lower bound on the value of the optimal policy, and hence the value of the option.  To price the option under the approximate policy, we proceed forward in time.  At each time step, the current exercise value is compared to the continuation value until either the option is exercised or the time horizon is reached.  The values for each sample path are then recorded and averaged to generate a lower bound on the true option value.

In previous literature, $\{C_t(X_{tj})\}_{j=1}^M$ has been estimated by regression splines for a single underlying asset~\citep{Ca96}, or least squares linear regression on a set of basis functions~\citep{TsVa99,LoSc01,Gl04}.  Regression on a set of basis functions becomes problematic when $X_{tj}$ is defined over moderate to high dimensional spaces. Well-defined sets of bases such as radial basis functions and polynomials require an exponential number of functions to span the space, while manually selecting basis functions can be quite difficult.  Since the expected continuation values are convex in the asset price for basket options, CAP is a simple, nonparametric alternative to these methods.

%Instead, we propose using CAP; the expected continuation values for a basket option are convex in the asset space.  The use of convex structure produces a better estimate over high-dimensional spaces than a naive regression method, while avoiding the pitfalls of selecting basis functions in moderate to high dimensions.

We compared the following methods: CAP and Fast CAP with $D = 3$, $L = 10$ for both and $P^{\prime} = \min(N,10)$; regression trees with constant leaves using the Matlab function {\tt classregtree}; least squares using the polynomial basis functions $$(1, S_i(t), S_i^2(t), S_i^3(t), S_i(t)S_j(t), h(X_t)), \ i = 1,\dots,N, \ j \neq i;$$ ridge regression on the same basis functions with ridge parameter chosen by 10-fold cross-validation each time period from values between $10^{-3}$ and $10^5$.

We compared value function regression methods as follows. We simulated on both $n = 10,000$ and $n = 50,000$ training samples for a 3-month American basket option with a number of underlying assets, $N$, varying between 1 and 30. Sample paths differed between $n = 10,000$ and $n = 50,000$. All asset sample paths were generated by a geometric Brownian motion with a drift of 0.05 and a volatility of 0.10. All assets had correlation 0.5 and starting value 100. The option had strike price 110. Policy values were approximated on 50,000 testing sample paths. An approximate upper bound was generated using the dual martingale methods of \citet{HaKo04} from value functions generated using polynomial basis functions based on the mean of the assets, $(1,Y,Y^2,Y^3,h(Y))$, where $Y_t = 1/N \sum_{i=1}^D X_{i}(t)$, with 5,000 samples. Approximate duality gaps were generated using these values and the policy for each method. All values are in discounted dollars. All computations were run in Matlab on a 2.66 GHz Intel i7 processor.

% 10-fold cross validation for LSR

\begin{table}
\footnotesize
\begin{center}
\begin{tabular}{c | l | r@{.}l | r@{.}l | c | r@{.}l | r@{.}l | r@{.}l  | r@{.}l}
{\bf Assets} & \multicolumn{1}{c|}{{\bf Method}} & \multicolumn{4}{c|}{{\bf Policy Value}} & {\bf Upper} & \multicolumn{4}{c|}{{\bf Duality Gap}} & \multicolumn{4}{c}{{\bf Time}} \\
\hline
& & \multicolumn{2}{c |}{n=10,000} & \multicolumn{2}{c |}{n=50,000} & & \multicolumn{2}{c |}{n=10,000} & \multicolumn{2}{c |}{n=50,000} & \multicolumn{2}{c |}{n=10,000} & \multicolumn{2}{c }{n=50,000}\\\hline
\multirow{5}{*}{{\bf 1}} & CAP & \  \ 27 & 6576 & \ \ {\bf27} & {\bf1496} & \multirow{5}{*}{30.2054} & \ \ \ \ \ \  8 & 4\% &  \ \ \ \ \ 10 & 1\% & 199 & 8 sec  & 25 & 0 min \\
& Fast CAP & {\bf27} & {\bf6667} & 27 & 0820 & & 8 & 4\% & 10 & 3\% & 16 & 8 sec & 184 & 4 sec \\
& Tree & 11 & 2055 & 11 & 0855 & & 62 & 9\% & 63 & 3\% & 30 & 7 sec & 272 & 9 sec \\
& LS & 27 & 5509 & 26 & 7066 &  & 8 & 8\% & 11 & 6\% & 0 & 1 sec & 0 & 7 sec  \\
& Ridge & 27 & 4661& 26 & 7804 &  & 9 & 1\% & 11 & 3\% & 25 & 4 sec & 160 & 0 sec \\
\hline
\multirow{5}{*}{{\bf 2}} & CAP & 24 & 0728 & {\bf24} & {\bf 0048} & \multirow{5}{*}{25.6928} & 6 & 3\% & 6 &  6\% & 288 & 5 sec & 34 & 2 min \\
& Fast CAP & {\bf24} & {\bf2015} & 23 & 9870 & & 5 & 8\% & 6 & 6\% & 21 & 8 sec & 211 & 0 sec \\
& Tree & 8 & 8716 & 9 & 0387 & & 65 & 5\% & 64 & 8\% & 37 & 8 sec & 319 & 6 sec \\
& LS & 23 & 6191 & 23 & 8338 &  & 8 & 1\% & 7 & 2\% & 0 & 3 sec & 1 & 3 sec \\
& Ridge & 23 & 5072 & 23 & 6710 &  & 8 & 5\% & 7 & 9\% & 39 & 2 sec & 219 & 8 sec  \\
\hline
\multirow{5}{*}{{\bf 5}} & CAP & {\bf20} & {\bf4537} & 20 & 0850 & \multirow{5}{*}{21.6261} & 5 & 5\% & 3 & 8\% & 818 & 6 sec & 78 & 2 min \\
& Fast CAP & 20 & 4287 & {\bf20} & {\bf9179} & & 5 & 5\% & 3 & 3\% & 89 & 2 sec & 522 & 6 sec \\
& Tree & 6 & 6613 & 6 & 4748 & & 69 & 2\% & 70 & 1\% & 60 & 0 sec & 438 & 5 sec \\
& LS & 19 & 7035 & 20 & 7115 &  & 8 & 9\% & 4 & 2\% & 1 & 1 sec & 8 & 4 sec \\
& Ridge & 20 & 0640 & 20 & 7249 &  & 7 & 2\% & 4 & 2\% & 116 & 8 sec & 493 & 9 sec \\
\hline
\multirow{5}{*}{{\bf 10}} & CAP & 19 & 0669 & \multicolumn{2}{c|}{--} & \multirow{5}{*}{21.0776} & 9 & 8\% &  \multicolumn{2}{c|}{--}  & 26 & 4 min &  \multicolumn{2}{c}{--}  \\
& Fast CAP & {\bf19} & {\bf1225} & {\bf19} & {\bf8889} & & 9 & 3\% & 5 & 6\% & 179 & 4 sec & 16 & 7 min \\
& Tree & 5 & 1832 & 5 & 4041 & & 75 & 4\% & 74 & 4\% & 92 & 1 sec & 641 & 7 sec \\
& LS & 17 & 5132 & 19 & 5932 &  & 16 & 9\% & 7 & 0\% & 8 & 6 sec & 58 & 8 sec \\
& Ridge & 18 & 9371 & 19 & 5198 & & 10 & 2\% & 7 & 4\% & 284 & 8 sec & 23 & 1 min \\
\hline
\multirow{5}{*}{{\bf 15}} & CAP & 18 & 4681 & \multicolumn{2}{c|}{--} & \multirow{5}{*}{21.1540} & 12 & 7\% & \multicolumn{2}{c|}{--}  & 49 & 2 min & \multicolumn{2}{c}{--}  \\
& Fast CAP & {\bf18} & {\bf4915} & {\bf19} & {\bf0486}  & & 12 & 6\% & 10 & 0\% & 212 & 2 sec & 22 & 8 min \\
& Tree & 4 & 7876 & 4  & 9438  & & 77 & 4\% & 76 & 6\% & 127 & 2 sec  & 826 & 3 sec \\
& LS & 14 & 5642 & 18 & 5433 &  & 31 & 2\% & 12 & 3\% & 32 & 1 sec & 170 & 7 sec\\
& Ridge & 18 & 1109 & 18 & 6270  &  & 14 & 4\% & 11 & 9\% & 852 & 9 sec & 59 & 9 min \\
\hline
\multirow{5}{*}{{\bf 20}} & CAP & 17 & 5874 & \multicolumn{2}{c|}{--} & \multirow{5}{*}{19.2050} & 8 & 4\% & \multicolumn{2}{c|}{--} & 75 & 6 min & \multicolumn{2}{c}{--} \\
& Fast CAP & {\bf18} & {\bf0322} & {\bf19} & {\bf3104} & & 6 & 1\% & -0 & 5\% & 267 & 1 sec & 26 & 4 min\\
& Tree & 4 & 7098 & 4 & 6530 & & 75 & 5\% & 75 & 8\% & 157 & 6 sec & 19 & 2 min \\
& LS & 11 & 7465 &  18 & 5712 & & 38 & 8\% & 3 & 3\% & 57 & 6 sec & 310 & 8 sec \\
& Ridge & 17 & 3077 & \multicolumn{2}{c|}{--}  &  & 9 & 9\% & \multicolumn{2}{c|}{--} & 28 & 6 min & \multicolumn{2}{c}{--} \\
\hline
\multirow{5}{*}{{\bf 30}} & CAP & 17 & 1366 & \multicolumn{2}{c|}{--} & \multirow{5}{*}{19.7415} & 13 & 2\%  & \multicolumn{2}{c|}{--} & 152 & 9 min & \multicolumn{2}{c}{--} \\
& Fast CAP & {\bf17} & {\bf3011} & {\bf 18} & {\bf5674}& & 12 & 4\% & 5 & 9\% & 339 & 3 sec & 44 & 3 min \\
& Tree & 4 & 4110 & 4 & 3111 & & 77 & 7\% & 78 & 2\% & 223 & 1 sec & 24 & 4 min \\
& LS & 7 & 3082 & 15 & 5700 &  & 63 & 0\% & 21 & 1\% & 169 & 9 sec & 21 & 1 min \\
& Ridge & 16 & 6310 & \multicolumn{2}{c|}{--} &  & 15 & 8\% & \multicolumn{2}{c|}{--} & 96 & 5 min & \multicolumn{2}{c}{--} \\
\hline
\end{tabular}
\end{center}
\normalsize
\caption{CAP, Fast CAP, tree regression with constant leaves, least squares (LS) and ridge regularized least squares were compared for pricing American basket options. Lower bounds were generated for each method by implementing the policy given by the value function. Computational times for each method are given in seconds. Approximate upper bounds were generated using \citet{HaKo04}. Duality gaps were calculated as a percentage of the approximate upper bound. The best lower bounds for each basket and sample size are bolded.}\label{tab:options}
\end{table}

% Bounds:
% - all probabilistic, upper bounds have much more variance due to construction

%The basket options were simple---all assets had the same drift, volatility, covariance and starting value---which should be one of the best scenarios for the basis functions used. However, CAP outperformed the other methods in all settings. As the number of assets grew, the vanilla least squares implementation began to overfit the data. Regularization removed most overfitting issues, but produced computational times similar to those of CAP. Additionally, it still required specification of basis functions.

% CAP and Fast CAP offer robust (and fast) performance without the hassles of basis functions

Results are displayed in Table \ref{tab:options}. We found that CAP and Fast CAP gave state of the art performance without the difficulties associated with linear functions, such as choosing basis functions and regularization parameters. We observed a decline in the performance of least squares as the number of assets grew due to overfitting. Ridge regularization greatly improved the least squares performance as the number of assets grew. Tree regression did poorly in all settings, likely due to overfitting in the presence of the non-symmetric error distribution generated by the geometric Brownian motion. These results suggest that CAP is robust even in less than ideal conditions, such as when data have heteroscedastic, non-symmetric error distributions.

Again, we noticed that while the performances of CAP and Fast CAP were comparable, the runtimes were about an order of magnitude different. On the larger problems, runtimes for Fast CAP were similar to those for unregularized least squares. This is likely because the number of covariates in the least squares regression grew like $N^2$, while all linear regressions in CAP only had $N$ covariates.

\section{Discussion}\label{sec:conclusions}
% - We give a new algorithm that is: computationally efficient, theoretically sound, empirically robust and really simple.
% - It is the first method that is computationally scalable for convex regression.
% - We hope that it allows new areas of study (econometrics problems with multiple dimensions, stochastic optimization with convex value-to-go/response surfaces, constraint approximation in convex optimization)
% - Future directions:
%	+ smooth/sparse variants for approximating objective functions that will be searched
% 	+ theoretical statements that explain empirical convergence rates

In this article, we presented Convex Adaptive Partitioning (CAP), a computationally efficient, theoretically sound and empirically robust method for regression subject to a convexity constraint. CAP is the first convex regression method to scale to large problems, both in terms of dimensions and number of observations. As such, we believe that it can allow the study of problems that were once thought to be computationally intractable. These include econometrics problems, like estimating consumer preference or production functions in multiple dimensions, approximating complex constraint functions for convex optimization, or creating convex value-to-go functions or response surfaces that can be easily searched in stochastic optimization. Our preliminary results are encouraging, but some important questions remain unanswered.
\begin{enumerate}
	\item What are the convergence rates for CAP? Are they adaptive, as they empirically seem to be?
	\item The current splitting proposal is effective but cumbersome. Are there less computationally intensive ways to refine the current partition?
	\item The modified stopping in Fast CAP provides substantially reduced runtimes with little performance degradation compared to CAP. Can this rule or a similarly efficient one be theoretically justified?
\end{enumerate}

We plan to explore this methodology further in the context of value function approximation, particularly in the situations where the value functions are searched as part of an objective function.

\bibliographystyle{agsm}

\end{document}